\documentclass[hidelinks,11pt]{article}

\usepackage{amsfonts,amsthm,xspace,graphicx,relsize,bm,mathtools,xcolor}
\usepackage{mathrsfs}
\usepackage[pagebackref]{hyperref}

\renewcommand{\backref}[1]{}

\renewcommand{\backrefalt}[4]{%
\ifcase #1
\or $^{#2}$%
\else $^{#2}$%
\fi}
\usepackage{kpfonts}
\usepackage[margin=1in]{geometry}
\hypersetup{
    colorlinks,
    linkcolor={blue!100!black},
    citecolor={blue!100!black},
}

\newtheorem{theorem}{Theorem} 

\newtheorem{lemma}[theorem]{Lemma}
\newtheorem{claim}[theorem]{Claim}
\newtheorem{fact}[theorem]{Fact}

\newcommand{\beq}{\begin{eqnarray}}
\newcommand{\eeq}{\end{eqnarray}}

\newcommand{\ket}[1]{|#1\rangle}
\newcommand{\bra}[1]{\langle#1|}

\newcommand{\ip}[2]{\langle #1 | #2 \rangle}
\newcommand{\ketbra}[2]{|#1\rangle\! \langle #2|}
\newcommand{\braketbra}[3]{\langle #1|#2| #3 \rangle}
\newcommand{\Tr}{\mbox{\rm Tr}}
\newcommand{\PEX}{\mbox{\rm PEX}}
\newcommand{\QPEX}{\mbox{\rm QPEX}}
\newcommand{\QAEX}{\mbox{\rm QAEX}}
\newcommand{\AEX}{\mbox{\rm AEX}}
\newcommand{\err}{\mbox{\rm err}}
\newcommand{\opt}{\mbox{\rm opt}}
\newcommand{\Id}{\ensuremath{\mathop{\rm Id}\nolimits}}

\DeclareMathOperator*{\Ex}{\mathbb{E}}

\def\01{\{0,1\}}
\newcommand{\C}{\ensuremath{\mathscr{C}}}
\newcommand{\G}{\ensuremath{\mathscr{G}}}

\newcommand{\A}{\ensuremath{\mathcal{A}}}

\newcommand{\Sh}{\ensuremath{\mathcal{S}}}

\newcommand{\F}{\ensuremath{\mathbb{F}}}

\newcommand{\R}{\ensuremath{\mathbb{R}}}

\newcommand{\Hi}{\ensuremath{\mathcal{H}}}
\newcommand{\E}{\mathcal{E}}
\newcommand{\M}{\mathcal{M}}

\newcommand{\supp}{\mathrm{supp}}

\newcommand{\eps}{\varepsilon}
\newcommand{\ph}{\ensuremath{\varphi}}

\begin{document}

\title{Optimal Quantum Sample Complexity\\ of Learning Algorithms}
\author{Srinivasan Arunachalam\thanks{QuSoft, CWI, Amsterdam, the Netherlands.
Supported by ERC Consolidator Grant QPROGRESS.}
\and
Ronald de Wolf\thanks{QuSoft, CWI and University of Amsterdam, the Netherlands.
Partially supported by ERC Consolidator Grant QPROGRESS.}
}
\date{}
\maketitle

\begin{abstract}
In learning theory, the \emph{VC dimension} of a concept class $\C$ is the most common way to measure its ``richness.''
A fundamental result says that the number of examples needed to learn an unknown target concept $c\in\C$ under an unknown distribution~$D$, is tightly determined by the VC dimension~$d$ of the concept class $\C$. Specifically, in the PAC model
$$
\Theta\Big(\frac{d}{\eps} + \frac{\log(1/\delta)}{\eps}\Big)
$$ 
examples are necessary and sufficient for a learner to output, with probability $1-\delta$, a hypothesis~$h$ that is $\eps$-close to the target concept~$c$ (measured under~$D$). In the related \emph{agnostic} model, where the samples need not come from a $c\in\C$, we know that 
$$
\Theta\Big(\frac{d}{\eps^2} + \frac{\log(1/\delta)}{\eps^2}\Big)
$$ 
examples are necessary and sufficient to output an hypothesis $h\in \C$ whose error is at most $\eps$ worse than the error of the best concept in $\C$.

Here we analyze \emph{quantum} sample complexity, where each example is a coherent quantum state. This model was introduced by Bshouty and Jackson~\cite{bshouty:quantumpac}, who showed that quantum examples are more powerful than classical examples in some fixed-distribution settings. However, At\i c\i\ and Servedio~\cite{atici&servedio:qlearning}, improved by Zhang~\cite{zhang:improvedvcbound}, showed that in the PAC setting (where the learner has to succeed for every distribution), quantum examples cannot be much more powerful: the required number of quantum examples is 
$$
\Omega\Big(\frac{d^{1-\eta}}{\eps} + d + \frac{\log(1/\delta)}{\eps}\Big)\mbox{ for arbitrarily small constant }\eta>0.
$$
Our main result is that quantum and classical sample complexity are in fact equal up to constant factors in both the PAC and agnostic models.
We give two proof approaches.  The first is a fairly simple information-theoretic argument that yields the above two classical bounds and yields the same bounds for quantum sample complexity up to a $\log(d/\eps)$ factor. We then give a second approach that avoids the log-factor loss, based on analyzing the behavior of the ``Pretty Good Measurement'' on the quantum  state identification problems that correspond to learning. This shows classical and quantum sample complexity are equal up to constant factors for every concept class~$\C$. 
 \end{abstract}

\section{Introduction}

\subsection{Sample complexity and VC dimension}

Machine learning is one of the most successful parts of AI, with impressive practical applications in areas ranging from image processing, speech recognition, to even beating Go champions.  Its theoretical aspects have been deeply studied, revealing beautiful structure and mathematical characterizations of when (efficient) learning is or is not possible in various~settings.

\subsubsection{The PAC setting}

Leslie Valiant's Probably Approximately Correct (PAC) model~\cite{valiant:paclearning} gives a precise complexity-theoretic definition of what it means for a concept class to be (efficiently) learnable. For simplicity we will (without loss of generality) focus on concepts that are Boolean functions, $c:\01^n\to\01$. Equivalently, a concept $c$ is a subset of $\01^n$, namely $\{x: c(x)=1\}$. Let $\C\subseteq\{f:\01^n\to\01\}$ be a concept class. This could for example be the class of functions computed by disjunctive normal form (DNF) formulas of a certain size, or Boolean circuits or decision trees of a certain depth.

The goal of a learning algorithm (the learner) is to probably approximate some unknown \emph{target concept} $c\in\C$ from random \emph{labeled examples}. Each labeled example is of the form $(x,c(x))$ where $x$ is distributed according to some unknown distribution~$D$ over $\01^n$.
After processing a number of such examples (hopefully not too many), the learner outputs some \emph{hypothesis}~$h$. We say that $h$ is \emph{$\eps$-approximately correct} (w.r.t.\ the target concept $c$)  if its error probability under~$D$ is at most~$\eps$: $\Pr_{x\sim D}[h(x)\neq c(x)]\leq\eps$. Note that the learning phase and the evaluation phase (i.e., whether a hypothesis is approximately correct) are according to the same distribution~$D$---as if the learner is taught and then tested by the same teacher. An $(\eps,\delta)$-learner for the concept class $\C$ is one whose hypothesis is probably approximately~correct:
\begin{quote}
For all target concepts $c\in\C$ and distributions $D$:\\ 
$\Pr[\mbox{the learner's output~$h$ is $\eps$-approximately correct}]\geq 1-\delta$,
\end{quote}
where the probability is over the sequence of examples and the learner's internal~randomness. Note that we leave the learner the freedom to output an $h$ which is not in $\C$. If always $h\in\C$, then the learner is called a \emph{proper} PAC-learner.

Of course, we want the learner to be as efficient as possible.  Its \emph{sample complexity} is the worst-case number of examples it uses, and its \emph{time complexity} is the worst-case running time of the learner. In this paper we focus on sample complexity.  This allows us to ignore technical issues of how the runtime of an algorithm is measured, and in what form the hypothesis $h$ is given as output by the learner.

The sample complexity of a concept class $\C$ is the sample complexity of the most efficient learner for $\C$. It is a function of $\eps$, $\delta$, and of course of $\C$ itself.  One of the most fundamental results in learning theory is that the sample complexity of $\C$ is tightly determined by a combinatorial parameter called the \emph{VC dimension} of $\C$, due to and named after Vapnik and Chervonenkis~\cite{vapnik:vcdimension}. The VC dimension of $\C$ is the size of the biggest $\Sh\subseteq\01^n$ that can be labeled in all $2^{|\Sh|}$ possible ways by concepts from $\C$: for each sequence of $|\Sh|$ binary labels for the elements of $\Sh$, there is a $c\in\C$ that has that labeling (such an $\Sh$ is said to be \emph{shattered} by $\C$).
Knowing this VC dimension (and $\eps,\delta$) already tells us the sample complexity of $\C$ up to constant factors.
Blumer et al.~\cite{blumer:optimalpacupper} proved that the sample complexity of $\C$ is lower bounded by $\Omega(d/\eps + \log(1/\delta)/\eps)$,
and they proved an upper bound that was worse by a $\log(1/\eps)$-factor.
In very recent work, Hanneke~\cite{hanneke:optimalpaclower} (improving on Simon~\cite{simon:almostoptimalpac}) got rid of this $\log(1/\eps)$-factor for PAC learning,\footnote{Hanneke's learner is not proper, meaning that its hypothesis $h$ is not always in~$\C$. It is still an open question whether the $\log(1/\eps)$-factor can be removed for proper PAC learning. Our lower bounds in this paper hold for all learners, quantum as well as classical, and proper as well as improper.} showing that the lower bound of Blumer et al.\ is in fact optimal: the sample complexity of $\C$ in the PAC setting~is
\begin{equation}\label{eq:samplecomplpac}
\Theta\Big(\frac{d}{\eps} + \frac{\log(1/\delta)}{\eps}\Big).
\end{equation}
\subsubsection{The agnostic setting}

The PAC model assumes that the labeled examples are generated according to a target concept $c\in\C$.
However, in many learning situations that is not a realistic assumption, for example when the examples are noisy in some way or when we have no reason to believe there is an underlying target concept at all. The \emph{agnostic} model of learning, introduced by Haussler~\cite{haussler:agnosticlearning} and Kearns et al.~\cite{kearns:agnosticlearning}, takes this into account.
Here, the examples are generated according to a distribution $D$ on $\01^{n+1}$. The error of a specific concept $c:\01^n\to\01$ is defined to be
$\err_D(c)=\Pr_{(x,b)\sim D}[c(x)\neq b]$.
When we are restricted to hypotheses in $\C$, we would like to find the hypothesis that minimizes $\err_D(c)$ over all $c\in\C$. However, it may require very many examples to do that exactly. In the spirit of the PAC model, the goal of the learner is now to output an $h\in\C$ whose error is at most an additive $\eps$ worse than that of the best ($=$ lowest-error) concepts in $\C$.  

Like in the PAC model, the optimal sample complexity of such agnostic learners is tightly determined by the VC dimension of~$\C$:~it is
\begin{equation}\label{eq:samplecomplagn}
\Theta\Big(\frac{d}{\eps^2} + \frac{\log(1/\delta)}{\eps^2}\Big),
\end{equation}
where the lower bound was proven by Vapnik and Chervonenkis~\cite{vapnik:agnosticlowerbound} (see also Simon~\cite{simon:agnosticlowerbound}), and the upper bound was proven by Talagrand~\cite{talagrand:agnosticupperbound}. Shalev-Shwartz and Ben-David~\cite[Section~6.4]{shwartz&david:learningbook} 
call Eq.~(\ref{eq:samplecomplpac}) and Eq.~(\ref{eq:samplecomplagn})
the ``Fundamental Theorem of PAC learning.''

\subsection{Our results}

In this paper we are interested in \emph{quantum} sample complexity. Here a \emph{quantum example} for some concept $c:\01^n\to\01$, according to some distribution~$D$, corresponds to an $(n+1)$-qubit state
$$
\sum_{x\in\01^n}\sqrt{D(x)}\ket{x,c(x)}.
$$ 
In other words, instead of a random labeled example, an example is now given by a coherent quantum superposition where the square-roots of the probabilities become the amplitudes.\footnote{We could allow more general quantum examples $\sum_{x\in\01^n}\alpha_x\ket{x,c(x)}$, where we only require $|\alpha_x|^2=D(x)$. However, that will not affect our results since our lower bounds apply to quantum examples where we know the amplitudes are square-rooted probabilities. Adding more degrees of freedom to quantum examples does not make learning~easier.}
This model was introduced by Bshouty and Jackson~\cite{bshouty:quantumpac}, who showed that DNF formulas are learnable in polynomial time from quantum examples when $D$ is uniform. For learning DNF under the uniform distribution from \emph{classical} examples, the best upper bound is quasipolynomial time~\cite{verbeurgt:learningdnf}. With the added power of ``membership queries,'' where the learner can actively ask for the label of any $x$ of his choice, DNF formulas are known to be learnable in polynomial time under uniform~$D$~\cite{jackson:dnf}, but \emph{without} membership queries polynomial-time learnability is a longstanding open~problem (see~\cite{daniely&shalevshwartz:limitdnf} for a recent hardness result). 

How reasonable are examples that are given as a coherent superposition rather than as a random sample?
They may seem unreasonable a priori because quantum superpositions seem very fragile and are easily collapsed by measurement, but if we accept the ``church of the larger Hilbert space'' view on quantum mechanics, where the universe just evolves unitarily without any collapses, then they may become more palatable. It is also possible that the quantum examples are generated by some coherent quantum process that acts like the teacher. 

How many quantum examples are needed to learn a concept class $\C$ of VC dimension~$d$?
Since a learner can just measure a quantum example in order to obtain a classical example, 
the \emph{upper} bounds on classical sample complexity trivially imply the same upper bounds on quantum sample complexity.
But what about the lower bounds?  Are there situations where quantum examples are more powerful than classical?  Indeed there are.  
We already mentioned the results of Bshouty and Jackson~\cite{bshouty:quantumpac} for learning DNF under the uniform distribution without membership queries.
Another good example is the learnability of the concept class of linear functions over $\mathbb{F}_2$, $\C=\{c(x)=a\cdot x: a\in\01^n\}$, again under the uniform distribution~$D$. It is easy to see that a classical learner needs about $n$ examples to learn an unknown $c\in\C$ under this~$D$. However, if we are given one quantum example
$$
\sum_{x\in\01^n}\sqrt{D(x)}\ket{x,c(x)}=\frac{1}{\sqrt{2^n}}\sum_{x\in\01^n}\ket{x,a\cdot x},
$$
then a small modification of the Bernstein-Vazirani algorithm~\cite{bernstein&vazirani:qcomplexity} can recover $a$ (and hence $c$) with probability 1/2. Hence $O(1)$ quantum examples suffice to learn $c$ exactly, with high probability, under the uniform distribution.  
At\i c\i\ and Servedio~\cite{atici&servedio:testing} used similar ideas to learning $k$-\emph{juntas} (concepts depending on only $k$ of their $n$ variables)
from quantum examples under the uniform distribution.
However, PAC learning requires a learner to learn $c$ under \emph{all possible} distributions~$D$, not just the uniform one. The success probability of the Bernstein-Vazirani algorithm deteriorates sharply when $D$ is far from uniform, but that does not rule out the existence of other quantum learners that use $o(n)$ quantum examples and succeed for all~$D$.

Our main result in this paper is that quantum examples are not actually more powerful than classical labeled examples in the PAC model and in the agnostic model: we prove that the lower bounds on classical sample complexity of Eq.~(\ref{eq:samplecomplpac}) and  Eq.~(\ref{eq:samplecomplagn}) hold for quantum examples as well.  Accordingly, despite several distribution-specific speedups, quantum examples do not significantly reduce sample complexity if we require our learner to work for all distributions~$D$. This should be contrasted with the situation when considering the \emph{time complexity} of learning. Servedio and Gortler~\cite{servedio&gortler:equivalencequantumclassical} considered a concept class (already known in the literature~\cite[Chapter~6]{kearns&valiant:blum}) that can be PAC-learned in polynomial time by a quantum computer, even with only classical examples, but that cannot be PAC-learned in polynomial time by a classical learner unless Blum integers can be factored in polynomial time (which is widely believed to be false).  

Earlier work on quantum sample complexity had already gotten close to extending the lower bound of Eq.~(\ref{eq:samplecomplpac}) to PAC learning from quantum examples. At\i c\i\ and Servedio~\cite{atici&servedio:qlearning} first proved a lower bound of $\Omega(\sqrt{d}/\eps + d + \log(1/\delta)/\eps)$ using the so-called ``hybrid method.'' Their proof technique was subsequently pushed further by Zhang~\cite{zhang:improvedvcbound} to
\begin{equation}\label{eq:previousbestbound}
\Omega\Big(\frac{d^{1-\eta}}{\eps} + d + \frac{\log(1/\delta)}{\eps}\Big)\mbox{ for arbitrarily small constant }\eta>0.
\end{equation}
Here we optimize these bounds, removing the $\eta$ and achieving the optimal lower bound for quantum sample complexity in the PAC model (Eq.~(\ref{eq:samplecomplpac})). 

We also show that the lower bound (Eq.~(\ref{eq:samplecomplagn})) for the agnostic model extends to quantum examples.
As far as we know, in contrast to the PAC model, no earlier results were known for quantum sample complexity in the agnostic model. 

We have two different proof approaches, which we sketch below.

\subsubsection{An information-theoretic argument}

In Section~\ref{section:infotheorylowerbounds} we give a fairly intuitive information-theoretic argument that gives optimal lower bounds for classical sample complexity, and that gives nearly-optimal lower bounds for quantum sample complexity. Let us first see how we can prove the classical PAC lower bound of Eq.~(\ref{eq:samplecomplpac}). Suppose $\Sh=\{s_0,s_1,\ldots,s_d\}$ is shattered by $\C$ (we now assume VC dimension $d+1$ for ease of notation). Then we can consider a distribution $D$ that puts probability $1-4\eps$ on $s_0$ and probability $4\eps/d$ on each of $s_1,\ldots,s_d$.\footnote{We remark that the distributions used here for proving lower bounds on quantum sample complexity have been used by Ehrenfeucht et al.~\cite{ehrenfeucht:lowerboundforlearning} for analyzing classical PAC sample~complexity.} For every possible labeling $(\ell_1 \ldots \ell_d)\in\01^d$ of $s_1,\ldots,s_d$ there will be a concept $c\in\C$ that labels $s_0$ with~0, and labels $s_i$ with $\ell_i$ for all $i\in\{1,\ldots,d\}$. Under $D$, most examples will be $(s_0,0)$ and hence give us no information when we are learning one of those $2^d$ concepts. Suppose we have a learner that $\eps$-approximates $c$ with high probability under this~$D$ using $T$ examples. Informally, our information-theoretic argument has the following three steps:
\begin{enumerate}
\item In order to $\eps$-approximate $c$, the learner has to learn the $c$-labels of at least $3/4$ of the $s_1,\ldots,s_d$ (since together these have $4\eps$ of the $D$-weight, and we want an $\eps$-approximation). As all $2^d$ labelings are possible, the $T$ examples together contain $\Omega(d)$ bits of information about $c$.
\item $T$ examples give at most $T$ times as much information about~$c$ as one example.
\item One example gives only $O(\eps)$ bits of information about $c$, because it will tell us one of the labels of $s_1,\ldots,s_d$ only with probability $4\eps$ (and otherwise it just gives $c(s_0)=0$).
\end{enumerate}
Putting these steps together implies $T=\Omega(d/\eps)$.\footnote{The other part of the lower bound of Eq.~(\ref{eq:samplecomplpac}) does not depend on $d$ and is fairly easy to prove.} 
This argument for the PAC setting is similar to an algorithmic-information argument of Apolloni and Gentile~\cite{apolloni&gentile:act} and an information-theoretic argument for variants of the PAC model with noisy examples of Gentile and~Helmbold~\cite{gentile&helmbold:it}. 

As far as we know, this type of reasoning has not yet been applied to 
the sample complexity of \emph{agnostic} learning. To get good lower bounds there, we consider a set of distributions $D_a$, indexed by $d$-bit string $a$. These distributions still have the property that if a learner gets $\eps$-close to the minimal error, then it will have to learn $\Omega(d)$ bits of information about the distribution (i.e., about~$a$). Hence the first step of the argument remains the same.  The second step of our argument also remains the same, and the third step shows an upper bound of $O(\eps^2)$ on the amount of information that the learner can get from one example.  This then implies $T=\Omega(d/\eps^2)$. We can also reformulate this for the case where we want the \emph{expected} additional error of the hypothesis over the best classifier in $\C$ to be at most~$\eps$, which is how lower bounds are often stated in learning theory.
We emphasize that our information-theoretic proof is simpler than the proofs in \cite{anthony&bartlett:learningbook,audibert:agnosticconstant1,shwartz&david:learningbook,aryeh:exactconstantagnostic}.

This information-theoretic approach recovers the optimal classical bounds on sample complexity, but also generalizes readily to the quantum case where the learner gets $T$ quantum examples. To obtain lower bounds on quantum sample complexity we use the same distributions~$D$ (now corresponding to a coherent quantum state) and basically just need to re-analyze the third step of the argument. In the PAC setting we show that one quantum example gives at most $O(\eps\log(d/\eps))$ bits of information about~$c$, and in the agnostic setting it gives $O(\eps^2\log(d/\eps))$ bits.
This implies lower bounds on sample complexity that are only a logarithmic factor worse than the optimal classical bounds for the PAC setting (Eq.~(\ref{eq:samplecomplpac})) and the agnostic setting (Eq.~(\ref{eq:samplecomplagn})).
This is not quite optimal yet, but already better than the previous best known lower bound (Eq.~(\ref{eq:previousbestbound})). 
The logarithmic loss in step~3 is actually inherent in this information-theoretic argument: in some cases a quantum example \emph{can} give roughly $\eps\log d$ bits of information about~$c$, for example when $c$ comes from the concept class of linear functions.

\subsubsection{A state-identification argument}

In order to get rid of the logarithmic factor we then try another proof approach, which views learning from quantum examples as a quantum state identification problem: we are given $T$ copies of the quantum example for some concept $c$ and need to $\eps$-approximate $c$ from this.
In order to render $\eps$-approximation of $c$ equivalent to exact identification of $c$, we use good linear error-correcting codes, restricting to concepts whose $d$-bit labeling of the elements of the shattered set $s_1,\ldots,s_d$ corresponds to a codeword.
We then have $2^{\Omega(d)}$ possible concepts, one for each codeword, and need to identify the target concept from a quantum state that is the tensor product of $T$ identical quantum~examples. 

State-identification problems have been well studied, and many tools are available for analyzing them.  In particular, we will use the so-called ``Pretty Good Measurement'' (PGM, also known as ``square root measurement''~\cite{hausladen:squareroot}) introduced by Hausladen and Wootters~\cite{hausladen:pgmintroduction}. The PGM is a specific measurement that one can always use for state identification, and whose success probability is no more than quadratically worse than that of the very best measurement.\footnote{Even better, in our application the PGM \emph{is} the optimal measurement, though this is not essential for our proof.} 
In Section~\ref{section:stateident} we use Fourier analysis to give an exact analysis of the average success probability of the PGM on the state-identification problems that come from both the PAC and the agnostic model. This analysis could be useful in other settings as well. Here it implies that the number of quantum examples, $T$, is lower bounded by Eq.~(\ref{eq:samplecomplpac}) in the PAC setting, and by Eq.~(\ref{eq:samplecomplagn}) in the agnostic setting.

Using the Pretty Good Measurement, we are also able to prove lower bounds for PAC learning under \emph{random classification noise}, which models the real-world situation that the learning data can have some errors. Classically in the random classification noise model (introduced by Angluin and Laird~\cite{angluin:randomclassificationnoise}), instead of obtaining labeled examples $(x,c(x))$ for some unknown $c\in \C$, the learner obtains \emph{noisy examples} $(x,b_x)$, where $b_x=c(x)$ with probability $1-\eta$ and $b_x=1-c(x)$ with probability $\eta$, for some \emph{noise rate} $\eta\in [0,1/2)$. 
Similarly, in the quantum learning model we could naturally define a \emph{noisy quantum example}~as~an~$(n+1)$-qubit~state 
$$
 \sum_{x\in\01^n}\sqrt{(1-\eta)D(x)}\ket{x,c(x)}+\sqrt{\eta D(x)}\ket{x,1-c(x)}.
$$ 
Using the PGM, we are able to show that the quantum sample complexity of PAC learning a concept class $\C$ under random classification noise~is:
\begin{equation}\label{eq:samplecomplnoisypac}
\Omega\Big(\frac{d}{(1-2\eta)^2\eps} + \frac{\log(1/\delta)}{(1-2\eta)^2\eps}\Big).
\end{equation}
We remark here that the best known classical sample complexity lower bound (see \cite{simon:agnosticlowerbound}) under the random classification noise is equal to the quantum sample complexity lower bound proven in Eq.~(\ref{eq:samplecomplnoisypac}).

\subsection{Related work}
Let us briefly discuss some related work in quantum learning theory, referring to our recent survey~\cite{arunachalam:quantumlearningsurvey} for more. In this paper we focus on \emph{sample} complexity, which is a fundamental information-theoretic quantity. Sample complexity concerns a form of ``passive'' learning: the learner gets a number of examples at the start of the process, and then has to extract enough information about the target concept from these.  We may also consider more active learning settings, in particular ones where the learner can make membership queries (i.e., learn the label $c(x)$ for any~$x$ of his choice). Servedio and Gortler~\cite{servedio&gortler:equivalencequantumclassical} showed that in this setting, classical and quantum complexity are polynomially related. They also exhibit an example of a factor-$n$ speed-up from quantum membership queries using the Bernstein-Vazirani~algorithm. Jackson et al.~\cite{jackson:quantumdnf} showed how quantum membership queries can improve Jackson's classical algorithm for learning DNF with membership queries under the uniform distribution~\cite{jackson:dnf}.

For \emph{quantum exact learning} (also referred to as the \emph{oracle identification} problem in the quantum literature), Kothari~\cite{kothari:oracleidentification} resolved a conjecture of Hunziker et al.~\cite{hunziker:quantumexactlearning}, that states that for any concept class $\C$, the number of quantum membership queries required  to exactly identify a concept $c\in \C$ is $O(\frac{\log |\C|}{\sqrt{\hat{\gamma}^\C}})$, where $\hat{\gamma}^\C$ is a combinatorial parameter of the concept class $\C$ which we shall not define here (see~\cite{atici&servedio:qlearning} for a precise definition). Montanaro~\cite{montanaro:learningpolynomials} showed how low-degree polynomials over a finite field can be identified more efficiently using quantum algorithms. 

In many ways the \emph{time} complexity of learning is at least as important as the sample complexity. We already mentioned that Servedio and Gortler~\cite{servedio&gortler:equivalencequantumclassical} exhibited a concept class based on factoring Blum integers that can be learned in quantum polynomial time but not in classical polynomial time, unless Blum integers can be factored efficiently.
Under the weaker (but still widely believed) assumption that one-way functions exist, they exhibited a concept class that can be learned exactly in polynomial time using quantum membership queries, but that takes superpolynomial time to learn from classical membership queries. Gavinsky~\cite{gavinsky:predictivelearning} introduced a model of learning called ``Predictive Quantum'' (PQ), a variation of quantum PAC learning, and exhibited a \emph{relational} concept class that is polynomial-time learnable in PQ, while any ``reasonable'' classical model requires an exponential number of classical examples to learn the concept~class.

A{\"{\i}}meur et al.~\cite{aimeur:mlinquantumworld,aimeur:qspeedup} consider a number of quantum algorithms in learning contexts such as clustering via minimum spanning tree, divisive clustering, and $k$-medians, using variants of Grover's algorithm~\cite{grover:search} to improve the time complexity of the analogous classical algorithms. Recently, there have been some quantum machine learning algorithms based on the HHL algorithm~\cite{hhl:lineq} for solving (in a weak sense) very well-behaved linear systems.  However, these algorithms often come with some fine print that limits their applicability, and their advantage over classical is not always clear.  We refer to Aaronson~\cite{aaronson:fineprint} for references and caveats. There has also been some work on quantum training of neural networks~\cite{wiebe:quantumdeeplearning,wiebe:quantumperceptronmodels}.

In addition to learning classical objects such as Boolean functions, one may also study the learnability of \emph{quantum} objects.  In particular, Aaronson~\cite{aaronson:qlearnability} studied how well $n$-qubit quantum states can be learned from measurement results.  In general, an $n$-qubit state $\rho$ is specified by $\exp(n)$ many parameters, and $\exp(n)$ measurement results on equally many copies of $\rho$ are needed to learn a good approximation of $\rho$ (say, in trace distance).  However, Aaronson showed an interesting and surprisingly efficient PAC-like result: from $O(n)$ measurement results, with measurements chosen i.i.d.\ according to an unknown distribution~$D$ on the set of all possible two-outcome measurements, we can learn an $n$-qubit quantum state~$\widetilde{\rho}$ that has roughly the same expectation value as $\rho$ for ``most'' possible two-outcome measurements. In the latter, ``most'' is again measured under~$D$, just like in the usual PAC learning the error of the learner's hypothesis is evaluated under the same distribution~$D$ that generated the learner's examples. Accordingly, $O(n)$ rather than $\exp(n)$ measurement results suffice to approximately learn an $n$-qubit state for most practical purposes.

The use of Fourier analysis in analyzing the success probability of the Pretty Good Measurement in quantum state identification appears in a number of earlier works. By considering the dihedral hidden subgroup problem (DHSP) as a state identification problem, Bacon et al.~\cite{bacon:hiddensubgroup} show that the PGM is the optimal measurement for DHSP and prove a lower bound on the sample complexity of $\Omega(\log |\G|)$ for a dihedral group $\G$ using Fourier analysis. Ambainis and Montanaro~\cite{ambainis&montanaro:wildcard&CGT} view the ``search with wildcard'' problem as a state identification problem. Using ideas similar to ours, they show that the $(x,y)$-th entry of the Gram matrix for the ensemble depends on the Hamming distance between $x$ and $y$, allowing them to use Fourier analysis to obtain an upper bound on the success probability of the state identification problem using the~PGM.

\subsection{Organization}
In Section~\ref{section:preliminaries} we formally define the classical and quantum learning models and introduce the Pretty Good Measurement. In Section~\ref{section:infotheorylowerbounds} we prove our information-theoretic lower bounds both for classical and quantum learning. In Section~\ref{section:stateident} we prove an optimal quantum lower bound for PAC and agnostic learning by viewing the learning process as a state identification problem. We conclude in Section~\ref{section:conclusion} with some open questions for further work.

\section{Preliminaries}
\label{section:preliminaries}

\subsection{Notation}
Let $[n]=\{1,\ldots,n\}$. For $x,y\in \01^d$, the bit-wise sum $x+ y$ is over $\F_2$, the \emph{Hamming distance} $d(x,y)$ is the number of indices on which $x$ and $y$ differ, $|x+y|$ is the Hamming weight of the string $x+y$ (which equals $d_H(x,y)$), and $x\cdot y=\sum_i x_iy_i$ (where the sum is over $\F_2$). For an $n$-dimensional vector space, the standard basis is denoted by $\{e_i\in \01^n:i\in [n]\}$, where $e_i$ is the vector with a $1$ in the $i$-th coordinate and $0$'s elsewhere. We write $\log$ for logarithm to base $2$, and $\ln$ for base~$e$. We will often use the bijection between the sets $\01^k$ and $[2^k]$ throughout this paper. Let $1_{[A]}$  be the indicator for an event $A$, and let $\delta_{x,y}=1_{[x=y]}$. We denote random variables in bold, such~as $\mathbf{A}$, $\mathbf{B}$.  

For a Boolean function $f:\01^m\rightarrow \01$ and $M\in \F_2^{m\times k}$ we define $f\circ M:\01^k\rightarrow \01$ as $(f\circ M)(x):=f(Mx)$ (where the matrix-vector product is over $\F_2$) for all $x\in \01^k$. For a distribution $D:\01^n\rightarrow [0,1]$, let $\supp(D)=\{x\in \01^n:D(x)\neq 0 \}$. By $x\sim D$, we mean $x$ is sampled according to the distribution $D$, i.e., $\Pr[\mathbf{X}=x]=D(x)$.

If $M$ is a positive semidefinite (psd) matrix, we define $\sqrt{M}$ as the unique psd matrix that satisfies $\sqrt{M}\cdot \sqrt{M}=M$, and $\sqrt{M}(i,j)$ as the $(i,j)$-th entry of $\sqrt{M}$. For a matrix $A\in \R^{m\times n}$, we denote the singular values of $A$ by $\sigma_1(A)\geq \sigma_2(A)\geq \cdots \geq \sigma_{\min\{m,n\}}(A)\geq 0$. The spectral norm of $A$ is $\|A\|=\max_{x\in \R^n,\|x\|=1}\|A x\|=\sigma_1$. Given a set of $d$-dimensional vectors $U=\{u_1,\ldots,u_n\}\in \R^{d}$, the Gram matrix $V$ corresponding to the set $U$ is the $n\times n$ psd matrix defined as $V(i,j)=u_i^t u_j$ for $i,j\in [n]$, where $u_i^t$ is the row vector that is the transpose of the column vector~$u_i$.

A technical tool used in our analysis of state identification problems is Fourier analysis on the Boolean cube. We will just introduce the basics of Fourier analysis here, referring to~\cite{odonnell:analysis} for more. Define the inner product between functions $f,g:\01^n\rightarrow \R$ as
$$
\langle f,g\rangle=\Ex_x [f(x)\cdot g(x)]
$$
where the expectation is uniform over $x\in \01^n$. For $S\subseteq [n]$ (equivalently $S\in\01^n$), let $\chi_S(x):=(-1)^{S\cdot x}$ denote the parity of the variables (of $x$) indexed by the set~$S$. It is easy to see that the set of functions $\{\chi_S\}_{S\subseteq [n]}$ forms an orthonormal basis for the space of real-valued functions over the Boolean cube. Hence every $f$ can be decomposed as 
$$
f(x)=\sum_{S\subseteq [n]} \widehat{f}(S) (-1)^{S\cdot x} \qquad \text{for all }x\in \01^n,
$$ 
where $\widehat{f}(S)=\langle f,\chi_S\rangle=\Ex_x [f(x)\cdot \chi_S(x)]$ is called a \emph{Fourier coefficient} of $f$.

\subsection{Learning in general}

In machine learning, a concept class $\C$ over $\01^n$ is a set of concepts $c:\01^n\rightarrow \01$.  We refer to a concept class $\C$ as being \emph{trivial} if either $\C$ contains only one concept, or $\C$ contains two concepts $c_0, c_1$ with $c_0(x)=1-c_1(x)$ for every $x\in \01^n$. For $c:\01^n\rightarrow \01$, we will often refer to the tuple $(x,c(x))\in \01^{n+1}$ as a \emph{labeled example}, where $c(x)$ is the~\emph{label}~of~$x$.

A central combinatorial concept in learning theory is the Vapnik-Chervonenkis (VC) dimension \cite{vapnik:vcdimension}. Fix a concept class $\C$ over $\01^n$. A set $\Sh=\{s_1,\ldots,s_t\}\subseteq \01^n$ is said to be \emph{shattered} by a concept class $\C$ if  $\{(c({s_1}),\ldots,c({s_t})) : c\in \C\} =\01^{t}$. In other words, for every labeling $\ell\in \01^{t}$, there exists a $c\in \C$ such that $(c({s_1}),\ldots,c({s_t}))=\ell$. The VC dimension of a concept class $\C$ is the  size of the largest $\Sh\subseteq \01^n$ that is shattered by~$\C$.

\subsection{Classical learning models}
\label{section:classicallearningmodels}
In this paper we will be concerned mainly with the PAC (Probably Approximately Correct) model of learning introduced by Valiant~\cite{valiant:paclearning}, and the agnostic model of learning introduced by Haussler~\cite{haussler:agnosticlearning} and Kearns et al.~\cite{kearns:agnosticlearning}. For further reading, see standard textbooks in computational learning theory such as~\cite{kearnss&vazirani:learningbook,anthony&bartlett:learningbook,shwartz&david:learningbook}.

In the classical PAC model, a learner $\A$ is given access to a \emph{random example oracle} $\PEX(c,D)$ which generates labeled examples of the form $(x,c(x))$ where $x$ is drawn from an unknown distribution $D:\01^n\rightarrow [0,1]$ and $c\in \C$ is the \emph{target concept} that $\A$ is trying to learn. For  a concept $c\in \C$ and hypothesis $h:\01^n\rightarrow \01$, we define the error of $h$ compared to the target concept $c$, under $D$, as $\err_D(h,c)=\Pr_{x\sim D}[h(x)\neq c(x)]$. A learning algorithm $\A$ is an \emph{$(\eps,\delta)$-PAC learner} for $\C$, if the following holds: 
\begin{quote}
For every $c\in\C$ and distribution $D$, given access to the $\PEX(c,D)$ oracle:\\ 
$\A$ outputs an $h$ such that $\err_D(h,c)\leq \eps$ with probability at least $1-\delta$. 
\end{quote}
The \emph{sample complexity} of $\A$ is the maximum number of invocations of the $\PEX(c,D)$ oracle which the learner makes, over all concepts $c\in \C$, distributions $D$, and the internal randomness of the learner. The \emph{$(\eps,\delta)$-PAC sample complexity} of a concept class $\C$ is the minimum sample complexity over all $(\eps,\delta)$-PAC learners for~$\C$. 

\emph{Agnostic} learning is the following model: for a distribution $D:\01^{n+1}\rightarrow [0,1]$, a learner $\A$ is given access to an $\AEX(D)$ oracle that generates examples of the form $(x,b)$ drawn from the distribution $D$.  We define the error of $h:\01^n\rightarrow \01$ under $D$ as $\err_D(h)=\Pr_{(x,b)\sim D}[h(x)\neq b]$. When $h$ is restricted to come from a concept class $\C$, 
the minimal error achievable is $\opt_D(\C)=\min_{c\in \C} \{\err_D(c)\}$. 
In agnostic learning, a learner $\A$ needs to output a hypothesis $h$ whose error is not much bigger than~$\opt_D(\C)$.
A learning algorithm $\A$ is an \emph{$(\eps,\delta)$-agnostic learner} for $\C$ if: 
\begin{quote}
For every distribution $D$ on $\01^{n+1}$, given access to the $\AEX(D)$~oracle:\\ 
$\A$ outputs an $h\in\C$ such that $\err_D(h)\leq \opt_D(\C)+\eps$ with probability at least $1-\delta$.
\end{quote}
Note that if there is a $c\in \C$ which perfectly classifies every $x$ with label~$y$ for $(x,y)\in \supp(D)$, then $\opt_D(\C)=0$ and we are in the setting of proper PAC learning. The \emph{sample complexity} of $\A$ is the maximum number of invocations of the $\AEX(c,D)$ oracle which the learner makes, over all distributions~$D$ and over the learner's internal randomness. The \emph{$(\eps,\delta)$-agnostic sample complexity} of a concept class $\C$ is the minimum sample complexity over all $(\eps,\delta)$-agnostic learners for~$\C$.

\subsection{Quantum information theory}
Throughout this paper we will assume the reader is familiar with the following quantum terminology. An $n$-dimensional \emph{pure state} is $\ket{\psi}=\sum_{i=1}^n \alpha_i \ket{i}$, where $\ket{i}$ is the $n$-dimensional unit vector that has a~1 only at position~$i$, the $\alpha_i$'s are complex numbers called the \emph{amplitudes}, and $\sum_{i\in [n]}|\alpha_i|^2=1$. An $n$-dimensional \emph{mixed state} (or \emph{density matrix}) $\rho=\sum_{i=1}^n p_i\ketbra{\psi_i}{\psi_i}$ is a mixture of pure states $\ket{\psi_1},\ldots,\ket{\psi_n}$ prepared with probabilities $p_1,\ldots,p_n$, respectively. The eigenvalues $\lambda_1,\ldots,\lambda_n$ of $\rho$ are non-negative reals and satisfy $\sum_{i\in [n]}\lambda_i=1$. If $\rho$ is pure (i.e., $\rho=\ketbra{\psi}{\psi}$ for some $\ket{\psi}$), then one of the eigenvalues is $1$ and the others are $0$.

To obtain classical information from $\rho$, one could apply a POVM (positive-operator-valued measure) to the state $\rho$. An $m$-outcome POVM is specified by a set of positive semidefinite  matrices $\{M_i\}_{i\in [m]}$ with the property $\sum_{i} M_i= \Id$. When this POVM is applied to the mixed state $\rho$, the probability of the $j$-th outcome is given by $\Tr(M_j\rho)$.

For a probability vector $(p_1,\ldots,p_k)$ (where $p_i\geq 0$ and $\sum_{i\in [k]} p_i=1$), the entropy function is defined as $H(p_1,\ldots,p_k)=-\sum_{i\in [k]} p_i\log p_i$. When $k=2$, with $p_1=p$ and $p_2=1-p$, we denote the binary entropy function as $H(p)$. For a state $\rho_{AB}$ on the Hilbert space $\Hi_A \otimes \Hi_B$, we let $\rho_A$ be the reduced state after taking the partial trace over $\Hi_B$.  The entropy of a quantum state $\rho_A$ is defined as $S(\mathbf{A})=-\Tr(\rho_{A} \log \rho_{A})$. The mutual information is defined as $I(\mathbf{A}:\mathbf{B})=S(\textbf{A})+S(\textbf{B})-S(\textbf{A}\textbf{B})$, and conditional entropy is defined as $S(\textbf{A}|\textbf{B})=S(\textbf{A}\textbf{B})-S(\textbf{B})$. Classical information-theoretic quantities correspond to the special case where $\rho$ is a diagonal matrix whose diagonal corresponds to the probability distribution of the random variable. Writing $\rho_A$ in its eigenbasis, it follows that $S(\mathbf{A})=H(\lambda_1,\ldots, \lambda_{\dim(\rho_A)})$, where $\lambda_1,\ldots,\lambda_{\dim(\rho_A)}$ are the eigenvalues of $\rho$. If $\rho_A$ is a pure state, $S(\mathbf{A})=0$.

\subsection{Quantum learning models}
The quantum PAC learning model was introduced by Bshouty and Jackson in \cite{bshouty:quantumpac}. The quantum PAC model is a generalization of the classical PAC model, instead of having access to random examples $(x,c(x))$ from the $\PEX(c,D)$ oracle, the learner now has access to superpositions over all~$(x,c(x))$. For an unknown distribution $D:\01^n\rightarrow [0,1]$ and concept $c\in \C$, a \emph{quantum example oracle} $\QPEX(c,D)$ acts on $\ket{0^n,0}$ and produces a \emph{quantum example} $\sum_{x\in \01^n} \sqrt{D(x)} \ket{x,c(x)}$ (we leave $\QPEX$ undefined on other basis states). A quantum learner is given access to some copies of the state generated by $\QPEX(c,D)$ and performs a POVM where each outcome is associated with a hypothesis. A learning algorithm $\A$ is an \emph{$(\eps,\delta)$-PAC quantum learner} for $\C$ if: 
 \begin{quote}
For every $c\in \C$ and distribution $D$, given access to the $\QPEX(c,D)$~oracle:\\ 
$\A$ outputs an $h$ such that $\err_D(h,c)\leq \eps$, with probability at least $1-\delta$. 
\end{quote}
The \emph{sample complexity} of $\A$ is the maximum number invocations of the $\QPEX(c,D)$ oracle, maximized over all $c\in \C$, distributions $D$, and the learner's internal randomness. The \emph{$(\eps,\delta$)-PAC quantum sample complexity} of a concept class $\C$ is the minimum sample complexity over all $(\eps,\delta)$-PAC quantum learners for~$\C$. 

We define quantum agnostic learning now. For a joint distribution $D:\01^{n+1}\rightarrow [0,1]$ over the set of examples, the learner has access to an $\QAEX(D)$  oracle which acts on $\ket{0^n,0}$ and produces a quantum example $\sum_{(x,b)\in \01^{n+1}} \sqrt{D(x,b)} \ket{x,b}$. A learning algorithm $\A$ is an \emph{$(\eps,\delta)$-agnostic quantum learner} for $\C$ if: 
\begin{quote}
For every distribution $D$, given access to the $\QAEX(D)$~oracle: \\
$\A$ outputs an $h\in\C$ such that $\err_D(h)\leq \opt_D(\C)+\eps$ with probability at least~$1-\delta$.
\end{quote}
The \emph{sample complexity} of $\A$ is the maximum number invocations of the $\QAEX(D)$ oracle over all distributions $D$ and over the learner's internal randomness. The \emph{$(\eps,\delta)$-agnostic quantum sample complexity} of a concept class $\C$ is the minimum sample complexity over all $(\eps,\delta)$-agnostic quantum learners for~$\C$.
  
\subsection{Pretty Good Measurement}
\label{section:pgm}
Consider an ensemble of $d$-dimensional states, $\E=\{(p_i,\ket{\psi_i})\}_{i\in[m]}$, where $\sum_{i\in [m]} p_i=1$. Suppose we are given an unknown state $\ket{\psi_{i}}$ sampled according to the probabilities and we are interested in maximizing the average probability of success to identify the state that we are given. For a POVM specified by positive semidefinite  matrices $\M=\{M_i\}_{i\in[m]}$, the probability of obtaining outcome $j$ equals $\braketbra{\psi_i}{M_j}{\psi_i}$. The average success probability is defined as
$$
P_{\M}(\E) = \sum_{i=1}^m p_i\braketbra{\psi_i}{M_i}{\psi_i}.
$$
Let $P^{opt}(\E)=\max_{\M} P_{\M}(\E)$ denote the optimal average success probability of $\E$, where the maximization is over the set of valid $m$-outcome POVMs.

For every ensemble $\E$, the so-called \emph{Pretty Good Measurement} (PGM) is a specific POVM (depending on the ensemble $\E$), which we shall define shortly, that does \emph{reasonably} well against $\E$. Suppose $P^{PGM}(\E)$ is defined as the average success probability of identifying the states in $\E$ using the PGM, then we have that 
$$
P^{opt}(\E)^2 \leq P^{PGM}(\E)\leq  P^{opt}(\E),
$$
 where the second inequality follows because $P^{opt}(\E)$ is a maximization over all valid POVMs and the first inequality was shown by Barnum and Knill~\cite{barnum:pgmmeasurement}. 

For completeness we give a simple proof of $P^{opt}(\E)^2 \leq P^{PGM}(\E)$ below (similar to~\cite{montanaro:distinguishability}). Let $\ket{\psi'_i}=\sqrt{p_i}\ket{\psi_i}$, and $\E'=\{\ket{\psi'_i} : i\in [m]\}$ be the set of states in~$\E$, renormalized to reflect their probabilities. Define $\rho=\sum_{i\in [m]} \ketbra{\psi'_i}{\psi'_i}$. The PGM is defined as the set of measurement operators $\{\ketbra{\nu_i}{\nu_i}\}_{i\in [m]}$ where $\ket{\nu_i}=\rho^{-1/2}\ket{\psi'_i}$ (the inverse square root of $\rho$ is taken over its non-zero eigenvalues). We first verify this is a valid POVM: 
$$
\sum_{i=1}^{m} \ketbra{\nu_i}{\nu_i}= \rho^{-1/2}\,\Big(\sum_{i=1}^{m} \ketbra{\psi'_i}{\psi'_i}\Big)\,\rho^{-1/2}=\Id.
$$
Let $G$ be the Gram matrix for the set $\E'$, i.e., $G(i,j)=\ip{\psi'_i}{\psi'_j}$ for $i,j\in [m]$. It can be verified that $\sqrt{G}(i,j)=\braketbra{\psi'_i}{\rho^{-1/2}}{\psi'_j}$. Hence 
\begin{align*}
P^{PGM}(\E)=\sum_{i\in [m]} p_i|\ip{\nu_i}{\psi_i}|^2&=\sum_{i\in [m]} |\ip{\nu_i}{\psi'_i}|^2\\
&=\sum_{i\in [m]} \braketbra{\psi'_i}{\rho^{-1/2}}{\psi'_i}^2=\sum_{i\in [m]} \sqrt{G}(i,i)^2.
\end{align*}
We now prove $P^{opt}(\E)^2 \leq P^{PGM}(\E)$. Suppose $\M$ is the optimal measurement. Since $\E$ consists of pure states, by a result of Eldar et al.~\cite{eldar:optimalquantumdetectors}, we can assume without loss of generality that the measurement operators in $\M$ are rank-1, so $M_i=\ketbra{\mu_i}{\mu_i}$ for some $\ket{\mu_i}$. 
Note that 
\begin{align}
\label{eq:upperboundmuirho}
\begin{aligned}
1=\Tr(\rho)&=\Tr\Big(\sum_{i\in[m]}\ketbra{\mu_i}{\mu_i}\rho^{1/2}\sum_{j\in[m]}\ketbra{\mu_j}{\mu_j}\rho^{1/2}\Big)\\
&=\sum_{i,j\in[m]}|\bra{\mu_i}\rho^{1/2}\ket{\mu_j}|^2\\
&\geq\sum_{i\in [m]}\bra{\mu_i}\rho^{1/2}\ket{\mu_i}^2.
\end{aligned}
\end{align}
Then, using the Cauchy-Schwarz inequality, we have
\begin{align*}
P^{opt}(\E)=\sum_{i\in [m]} |\ip{\mu_i}{\psi'_i}|^2&=\sum_{i\in [m]} |\braketbra{\mu_i}{\rho^{1/4}\rho^{-1/4}}{\psi'_i}|^2 \\
&\leq \sum_{i\in [m]} \braketbra{\mu_i}{\rho^{1/2}}{\mu_i}\braketbra{\psi'_i}{\rho^{-1/2}}{\psi'_i} \\
&\leq \sqrt{ \sum_{i\in [m]}  \braketbra{\mu_i}{\rho^{1/2}}{\mu_i}^2}\sqrt{ \sum_{i\in [m]} \braketbra{\psi'_i}{\rho^{-1/2}}{\psi'_i}^2} \\
&\stackrel{\text{Eq.}~(\ref{eq:upperboundmuirho})}{\leq} \sqrt{ \sum_{i\in [m]} \braketbra{\psi'_i}{\rho^{-1/2}}{\psi'_i}^2}\\
&= \sqrt{P^{PGM}(\E)}.
\end{align*}

The above shows that for all ensembles~$\E$, the PGM for that ensemble is not much worse than the optimal measurement. In some cases the PGM \emph{is} the optimal measurement.  In particular, an ensemble $\E$ is called \emph{geometrically uniform} if $\E=\{U_i\ket{\ph}:i\in [m]\}$ for some Abelian group of matrices $\{U_i\}_{i\in [m]}$ and state~$\ket{\ph}$. Eldar and Forney~\cite{eldar&forney:squarerootmeasurement} showed $P^{opt}(\E)=P^{PGM}(\E)$ for such~$\E$.

\subsection{Known results and required claims}
The following theorems characterize the sample complexity of classical PAC and~agnostic~learning.

\begin{theorem} [\cite{blumer:optimalpacupper,hanneke:optimalpaclower}]
\label{thm:classicalpaclearning}
Let $\C$ be a concept class with VC-dim$(\C)=d+1$.
In the PAC model, $\Theta\Big(\frac{d}{\eps} + \frac{\log(1/\delta)}{\eps}\Big) $ examples are necessary and sufficient for a classical $(\eps,\delta)$-PAC learner for~$\C$.
\end{theorem}

\begin{theorem} [\cite{vapnik:agnosticlowerbound,simon:agnosticlowerbound, talagrand:agnosticupperbound}]
\label{thm:classicalagnosticlearning}
Let $\C$ be a concept class with VC-dim$(\C)=d$.
In the agnostic model, $\Theta\Big(\frac{d}{\eps^2} + \frac{\log(1/\delta)}{\eps^2}\Big) $ examples are necessary and sufficient for a classical $(\eps,\delta)$-agnostic learner for $\C$.
\end{theorem}

We will use the following well-known theorem from the theory of error-correcting codes: 
\begin{theorem}
\label{thm:randomlinearcode}
For every sufficiently large integer $n$, there exists an integer $k\in[n/4,n]$ and a matrix $M\in\mathbb{F}_2^{n\times k}$ of rank $k$, such that the associated $[n,k,d]_2$ linear code $\{Mx:x\in\01^k\}$ has minimal distance~$d\geq~n/8$.
\end{theorem}

We will need the following claims later

\begin{claim}
\label{claim:Fouriercoeffmatrixprod}
Let $f:\01^m\rightarrow \R$ and let $M\in \F_2^{m\times k}$. Then the Fourier coefficients of $f\circ M$ are $\widehat{f\circ M}(Q)=\sum_{S\in \01^m:M^t S =Q} \widehat{f} (S)$ for all $Q\subseteq [k]$ (where $M^t$ is the transpose of the matrix $M$).  
\end{claim}
\begin{proof}
Writing out the Fourier coefficients of $f\circ M$
\begin{align*}
\widehat{f\circ M}(Q)&= \Ex_{z\in \01^k} [ (f\circ M)(z) (-1)^{Q\cdot z}] \\
&=\Ex_{z\in \01^k} \Big[\sum_{S\in \01^m} \widehat{f}(S) (-1)^{S\cdot (M z)+Q\cdot z}\Big]  \tag{Fourier expansion of $f$}\\
&=\sum_{S\in \01^m} \widehat{f}(S) \Ex_{z\in \01^k} [(-1)^{(M^t S+Q)\cdot  z}] \tag{using $\langle S,Mz\rangle =\langle M^t S,z\rangle $} \\
&=\sum_{S:M^t S =Q} \widehat{f} (S) \tag{using $\Ex_{z\in \01^k} (-1)^{(z_1+z_2)\cdot z}=\delta_{z_1,z_2}$}.
\end{align*}
\end{proof}
\begin{claim}
\label{claim:maximizingc/sqrtt}
$\max \{(c/\sqrt{t})^t: t\in [1,c^2]\}=e^{c^2/(2e)}$.
\end{claim}
\begin{proof}
The value of $t$ at which the function $\Big(c/\sqrt{t}\Big)^t$ is the largest, is obtained by differentiating the function with respect to $t$,
$$
\frac{d}{dt} \Big(c/{\sqrt{t}}\Big)^t= (c/\sqrt{t})^t \Big(\ln (c/\sqrt{t})-1/2\Big).
$$
Equating the derivative to zero we obtain the maxima (the second derivative can be checked to be negative) at $t=c^2/e$.
\end{proof}
\begin{fact}
\label{fact:taylorseriesbinaryentropy}
For all $\eps \in [0,1/2]$ we have $H(\eps)\leq O(\eps \log (1/\eps))$, and (from the Taylor series) 
$$
1-H(1/2+\eps)\leq 2\eps^2/\ln 2+O(\eps^4).
$$
\end{fact}
\begin{fact}
\label{fact:binomialupperbound}
For every positive integer $n$, we have that ${n\choose k}\leq 2^{nH(k/n)}$ for all $k\leq n$  and $\sum_{i=0}^{m}{n\choose i}\leq 2^{nH(m/n)}$ for all $m\leq n/2$. 
\end{fact}

The following facts are well-known in quantum information theory.
\begin{fact}
\label{fact:2statedistinguishingbound}
Let binary random variable $\mathbf{b}\in\01$ be uniformly distributed. Suppose an algorithm is given $\ket{\psi_{\mathbf{b}}}$ (for unknown $b$) and is required to guess whether $\mathbf{b}=0$ or $\mathbf{b}=1$. It will guess correctly with probability at most $\frac{1}{2}+\frac{1}{2}\sqrt{1-|\ip{\psi_0}{\psi_1}|^2}$.
\end{fact}

Note that if we could distinguish between the states $\ket{\psi_0}$ and $\ket{\psi_1}$ with probability $\geq 1-\delta$, then~$|\ip{\psi_0}{\psi_1}|\leq 2\sqrt{\delta(1-\delta)}$.

\begin{fact}
\label{fact:quantumsubadditivity}
(Subadditivity of quantum entropy): For an arbitrary bipartite state $\rho_{AB}$ on the Hilbert space $\Hi_A\otimes \Hi_B$, it holds that $S(\rho_{AB})\leq S(\rho_A)+S(\rho_B)$.
\end{fact}


\section{Information-theoretic lower bounds}
\label{section:infotheorylowerbounds}

Upper bounds on sample complexity carry over from classical to quantum PAC learning, because a quantum example becomes a classical example if we just measure it. Our main goal is to show that the \emph{lower} bounds also carry over. All our lower bounds will involve two terms, one that is independent of  $\C$ and one that is dependent on the VC dimension of $\C$. In Section~\ref{section:dindependentpart} we prove the VC-independent part of the lower bounds for the \emph{quantum} setting (which also is a lower bound for the classical setting), in Section~\ref{section:infortheoreticlowerboundsclassical} we present an information-theoretic lower bound on sample complexity for PAC learning and agnostic learning which yields optimal VC-dependent bounds in the classical case. Using similar ideas, in Section~\ref{section:infotheoreticlowerboundsquantum} we obtain near-optimal bounds in the quantum case.

\subsection{VC-independent part of lower bounds}
\label{section:dindependentpart}
\begin{lemma} [\cite{atici&servedio:qlearning}]
\label{lemma:deltapartoflemmapac}
Let $\C$ be a non-trivial concept class.
For every $\delta\in (0,1/2)$, $\eps\in (0,1/4)$, a $(\eps,\delta)$-PAC quantum learner for~$\C$ has sample complexity $\Omega(\frac{1}{\eps}\log \frac{1}{\delta})$.
\end{lemma}

\begin{proof}
Since $\C$ is non-trivial, we may assume there are two concepts $c_1,c_2 \in \C$ defined on two inputs $\{x_1,x_2\}$ as follows $c_1(x_1)=c_2(x_1)=0$ and $c_1(x_2)=0, c_2(x_2)=1$. Consider the distribution $D(x_1)=1-\eps$ and $D(x_2)=\eps$. For $i\in \{1,2\}$, the state of the algorithm after $T$ queries to $\QPEX(c_i,D)$ is $\ket{\psi_i}=(\sqrt{1-\eps}\ket{x_1,0}+\sqrt{\eps}\ket{x_2,c_i(x_2)})^{\otimes T}$. It follows that $\ip{\psi_1}{\psi_2}=(1-\eps)^{T}$. Since the success probability of an $(\eps,\delta)$-PAC quantum learner is $\geq 1-\delta$, Fact~\ref{fact:2statedistinguishingbound} implies $\ip{\psi_1}{\psi_2}\leq 2\sqrt{\delta(1-\delta)}$. Hence $T=\Omega(\frac{1}{\eps}\log \frac{1}{\delta})$.
\end{proof}

\begin{lemma}
\label{lemma:deltapartoflemmaagnostic}
Let $\C$ be a non-trivial concept class.
For every $\delta\in (0,1/2)$, $\eps\in (0,1/4)$, a $(\eps,\delta)$-agnostic quantum learner for~$\C$ has sample complexity $\Omega(\frac{1}{\eps^2}\log \frac{1}{\delta})$.
\end{lemma}

\begin{proof}
Since $\C$ is non-trivial, we may assume there are two concepts $c_1,c_2 \in \C$ and there exists an input $x\in \01^n$ such that $c_1(x)\neq c_2(x)$. Consider the two distributions $D_-$ and $D_+$ defined as follows: $D_{\pm}(x,c_1(x))=(1\pm \eps)/2$ and $D_{\pm}(x,c_2(x))=(1\mp \eps)/2$. Let $\ket{\psi_{\pm}}$ be the state after $T$ queries to $\QAEX(D_{\pm})$, i.e., $\ket{\psi_{\pm}}=(\sqrt{(1\pm \eps)/2}\ket{x,c_1(x)}+\sqrt{(1\mp \eps)/2}\ket{x,c_2(x)})^{\otimes T}$. It follows that $\ip{\psi_+}{\psi_-}=(1-\eps^2)^{T/2}$. 
Since the success probability of an $(\eps,\delta)$-agnostic quantum learner is $\geq 1-\delta$, Fact~\ref{fact:2statedistinguishingbound} implies $\ip{\psi_+}{\psi_-}\leq 2\sqrt{\delta(1-\delta)}$. Hence~$T=\Omega(\frac{1}{\eps^2}\log \frac{1}{\delta})$ 
\end{proof}

\subsection{Information-theoretic lower bounds on sample complexity: classical case}
\label{section:infortheoreticlowerboundsclassical}

\subsubsection{Optimal lower bound for  classical PAC learning}

\begin{theorem}\label{thm:infoclassicalpac}
Let $\C$ be a concept class with VC-dim$(\C)=d+1$. Then for every $\delta\in (0,1/2)$ and $\eps \in (0,1/4)$, every $(\eps,\delta)$-PAC learner for~$\C$ has sample complexity 
$\Omega\Big(\frac{d}{\eps} + \frac{\log(1/\delta)}{\eps}\Big)$.
\end{theorem}

\begin{proof}
Consider an $(\eps,\delta)$-PAC learner for~$\C$ that uses $T$ examples.
The $d$-independent part of the lower bound, $T=\Omega(\log(1/\delta)/\eps)$, even holds for quantum examples and was proven in Lemma~\ref{lemma:deltapartoflemmapac}.
Hence it remains to prove $T=\Omega(d/\eps)$.
It suffices to show this for a specific distribution~$D$, defined as follows. Let $\Sh=\{s_0,s_1,\ldots,s_d\} \subseteq \01^n$ be some $(d+1)$-element set shattered by $\C$. Define $D(s_0)=1-4\eps$ and $D(s_i)=4\eps/d$ for all $i\in [d]$. 

Because $\Sh$ is shattered by $\C$, for each string $a\in\01^d$, there exists a concept $c_a\in\C$ such that $c_a(s_0)=0$ and $c_a(s_i)=a_i$ for all $i\in[d]$. We define two correlated random variables $\mathbf{A}$ and $\mathbf{B}$ corresponding to the concept and to the examples, respectively. Let $\mathbf{A}$ be a random variable that is uniformly distributed over $\01^d$; if $\mathbf{A}=a$, let $\mathbf{B}=\mathbf{B}_1\ldots\mathbf{B}_T$ be $T$ i.i.d.\ examples from $c_a$ according to $D$. We give the following three-step analysis of these random variables:
\begin{enumerate}
\item $I(\mathbf{A}:\mathbf{B})\geq (1-\delta)(1-H(1/4))d - H(\delta)=\Omega(d)$.\\[1mm]
\emph{Proof.} Let random variable $h(\mathbf{B})\in\01^d$ be the hypothesis that the learner produces (given the examples in $\mathbf{B}$) restricted to the elements $s_1,\ldots,s_d$. Note that the error of the hypothesis $\err_D(h(\mathbf{B}),c_{\mathbf{A}})$ equals $ d_H(\mathbf{A},h(\mathbf{B}))\cdot 4\eps/d$, because each $s_i$ where $\mathbf{A}$ and $h(\mathbf{B})$ differ contributes $D(s_i)=4\eps/d$ to the error. Let $\mathbf{Z}$ be the indicator random variable for the event that the error is $\leq\eps$. If $\mathbf{Z}=1$, then $d_H(\mathbf{A},h(\mathbf{B}))\leq d/4$. Since we are analyzing an $(\eps,\delta)$-PAC learner, we have $\Pr[\mathbf{Z}=1]\geq 1-\delta$, and $H(\mathbf{Z})\leq H(\delta)$.  Given a string $h(\mathbf{B})$ that is $d/4$-close to $\mathbf{A}$, $\mathbf{A}$ ranges over a set of only $\sum_{i=0}^{d/4}{d\choose i}\leq 2^{H(1/4)d}$ possible $d$-bit strings (using Fact~\ref{fact:binomialupperbound}), hence $H(\mathbf{A}\mid\mathbf{B},\mathbf{Z}=1)\leq H(\mathbf{A}\mid h(\mathbf{B}),\mathbf{Z}=1)\leq H(1/4)d$.
We now lower bound $I(\mathbf{A}:\mathbf{B})$ as follows:
\begin{align*}
I(\mathbf{A}:\mathbf{B}) & = H(\mathbf{A})-H(\mathbf{A}\mid\mathbf{B})\\
       &  \geq H(\mathbf{A})-H(\mathbf{A}\mid \mathbf{B},\mathbf{Z})-H(\mathbf{Z})\\ 
       & = H(\mathbf{A})-\Pr[\mathbf{Z}=1]\cdot H(\mathbf{A}\mid \mathbf{B},\mathbf{Z}=1)-\Pr[\mathbf{Z}=0]\cdot H(\mathbf{A}\mid\mathbf{B},\mathbf{Z}=0) - H(\mathbf{Z})\\
       & \geq d-(1-\delta)H(1/4)d - \delta d - H(\delta)\\
       & = (1-\delta)(1-H(1/4))d - H(\delta).
\end{align*}
\item $I(\mathbf{A}:\mathbf{B})\leq T\cdot I(\mathbf{A}:\mathbf{B}_1)$.\\[1mm]
\emph{Proof.} This inequality is essentially due to Jain and Zhang~\cite[Lemma~5]{jain&zhang:newoneway}, we include the proof for completeness.
\begin{align*}
I(\mathbf{A}:\mathbf{B}) = H(\mathbf{B})- H(\mathbf{B}\mid \mathbf{A}) 
       &= H(\mathbf{B})- \sum_{i=1}^T H(\mathbf{B}_i\mid \mathbf{A})\\
       &\leq \sum_{i=1}^T H(\mathbf{B}_i) - \sum_{i=1}^T H(\mathbf{B}_i\mid \mathbf{A})
       = \sum_{i=1}^T I(\mathbf{A}:\mathbf{B}_i),
\end{align*}
where the second equality used independence of the $\mathbf{B}_i$'s conditioned on $\mathbf{A}$, and the inequality uses Fact~\ref{fact:quantumsubadditivity}.
Since $I(\mathbf{A}:\mathbf{B}_i)=I(\mathbf{A}:\mathbf{B}_1)$ for all~$i$, we get the inequality.
\item $I(\mathbf{A}:\mathbf{B}_1)=4\eps$.\\[1mm]
\emph{Proof.} View $\mathbf{B}_1=(\mathbf{I},\mathbf{L})$ as consisting of an index $\mathbf{I}\in\{0,1,\ldots,d\}$ and a corresponding label $\mathbf{L}\in\01$. With probability $1-4\eps$, $(\mathbf{I},\mathbf{L})=(0,0)$. For each $i\in[d]$, with probability $4\eps/d$, $(\mathbf{I},\mathbf{L})=(i,\mathbf{A}_i)$. Note that $I(\mathbf{A}:\mathbf{I})=0$ because $\mathbf{I}$ is independent of~$\mathbf{A}$; $I(\mathbf{A}:\mathbf{L}\mid\mathbf{I}=0)=0$; and $I(\mathbf{A}:\mathbf{L}\mid \mathbf{I}=i)=I(\mathbf{A}_i:\mathbf{L}\mid\mathbf{I}=i)=H(\mathbf{A}_i\mid\mathbf{I}=i)-H(\mathbf{A}_i\mid \mathbf{L},\mathbf{I}=i)=1-0=1$ for all $i\in[d]$. We~have 
$$
I(\mathbf{A}:\mathbf{B}_1)=I(\mathbf{A}:\mathbf{I})+I(\mathbf{A}:\mathbf{L}\mid\mathbf{I})=\sum_{i=1}^d \Pr[\mathbf{I}=i]\cdot I(\mathbf{A}:\mathbf{L}\mid\mathbf{I}=i)=4\eps.
$$
\end{enumerate}
Combining these three steps implies $T=\Omega(d/\eps)$.
\end{proof}

\subsubsection{Optimal lower bound for classical agnostic learning}

\begin{theorem}\label{thm:infoclassicalagn}
Let $\C$ be a concept class with VC-dim$(\C)=d$. Then for every $\delta\in (0,1/2)$ and $\eps \in (0,1/4)$, every $(\eps,\delta)$-agnostic learner for~$\C$ has sample complexity 
$\Omega\Big(\frac{d}{\eps^2} + \frac{\log(1/\delta)}{\eps^2}\Big)$.
\end{theorem}

\begin{proof}
The $d$-independent part of the lower bound, $T=\Omega(\log(1/\delta)/\eps^2)$, even holds for quantum examples and was proven in Lemma~\ref{lemma:deltapartoflemmaagnostic}. For the other part, the proof is similar to Theorem~\ref{thm:infoclassicalpac}, as follows.
Assume an $(\eps,\delta)$-agnostic learner for $\C$ that uses $T$ examples.
We need to prove $T=\Omega(d/\eps^2)$.
For shattered set $\Sh=\{s_1,\ldots,s_d\} \subseteq \01^n$ and $a\in\01^d$, define distribution $D_a$ on $[d]\times\01$ by $D_a(i,\ell)=(1+(-1)^{a_i+\ell}4\eps)/2d$. 

Again let random variable $\mathbf{A}\in\01^d$ be a uniformly distributed random variable, corresponding to the values of concept $c_a$ on $\Sh$, and $\mathbf{B}=\mathbf{B}_1\ldots\mathbf{B}_T$ be $T$ i.i.d.\ samples from~$D_a$. Note that $c_a$ is the minimal-error concept from $\C$ w.r.t.\ $D_a$, and concept $c_{\tilde{a}}$ has additional error $d_H(a,\tilde{a})\cdot 4\eps/d$.
Accordingly, an $(\eps,\delta)$-agnostic learner has to produce (from $\mathbf{B}$) an $h(\mathbf{B})\in\01^d$, which, with probability at least $1-\delta$, is $d/4$-close to~$\mathbf{A}$.
Our three-step analysis is very similar to Theorem~\ref{thm:infoclassicalpac};
only the third step changes:
\begin{enumerate}
\item $I(\mathbf{A}:\mathbf{B})\geq (1-\delta)(1-H(1/4))d - H(\delta)=\Omega(d)$.
\item $I(\mathbf{A}:\mathbf{B})\leq T\cdot I(\mathbf{A}:\mathbf{B}_1)$.
\item $I(\mathbf{A}:\mathbf{B}_1)=1-H(1/2+2\eps)=O(\eps^2)$.\\[1mm]
\emph{Proof.} View the $D_a$-distributed random variable $\mathbf{B}_1=(\mathbf{I},\mathbf{L})$ as index $\mathbf{I}\in[d]$ and label $\mathbf{L}\in\01$. The marginal distribution of $\mathbf{I}$ is uniform; conditioned on $\mathbf{I}=i$, the bit $\mathbf{L}$ equals $\mathbf{A}_i$ with probability $1/2+2\eps$. Hence 
$$
I(\mathbf{A}:\mathbf{L}\mid \mathbf{I}=i)=I(\mathbf{A}_i:\mathbf{L}\mid \mathbf{I}=i)
=H(\mathbf{A}_i\mid \mathbf{I}=i)-H(\mathbf{A}_i\mid\mathbf{L},\mathbf{I}=i)=1-H(1/2+2\eps).
$$ 
Using Fact~\ref{fact:taylorseriesbinaryentropy}, we have
\begin{align*}
I(\mathbf{A}:\mathbf{B}_1)=I(\mathbf{A}:\mathbf{I})+I(\mathbf{A}:\mathbf{L}\mid \mathbf{I})&=\sum_{i=1}^d \Pr[\mathbf{I}=i]\cdot I(\mathbf{A}:\mathbf{L}\mid \mathbf{I}=i)\\
&=1-H(1/2+2\eps)=O(\eps^2).
\end{align*}
\end{enumerate}
Combining these three steps implies $T=\Omega(d/\eps^2)$.
\end{proof}

In the theorem below, we optimize the constant in the lower bound of the sample complexity in Theorem~\ref{thm:infoclassicalagn}. 
In learning theory such lower bounds are often stated slightly differently. In order to compare the lower bounds, we introduce the following. We first define an \emph{$\eps$-average agnostic learner} for a concept class $\C$ as a learner that, given access to $T$ samples from an $\AEX(D)$ oracle (for some unknown distribution $D$), needs to output a hypothesis $h_{\mathbf{X}\mathbf{Y}}$ (where ${(\mathbf{X},\mathbf{Y})\sim D^T}$) that satisfies 
$$
\Ex_{(\mathbf{X},\mathbf{Y})\sim D^T} [ \err_D(h_{\mathbf{X}\mathbf{Y}})]- \opt_D(\C)\leq \eps.
$$
Lower bounds on the quantity $(\Ex_{(\mathbf{X},\mathbf{Y})\sim D^T} [ \err_D(h_{\mathbf{X}\mathbf{Y}})]- \opt_D(\C))$ are generally referred to as \emph{minimax lower bounds} in learning theory. For concept class $\C$, Audibert~\cite{audibert:agnosticconstant2,audibert:agnosticconstant1} showed that there exists a distribution~$D$, such that if the agnostic learner uses $T$ samples from $\AEX(D)$, then
$$
\Ex_{(\mathbf{X},\mathbf{Y})\sim D^T} [\err_D(h_{\mathbf{X}\mathbf{Y}})]- \opt_D(\C)\geq \frac{1}{6}\sqrt{\frac{d}{T}}.
$$
Equivalently, this is a lower bound of $T\geq \frac{d}{36\eps^2}$ on the sample complexity of an $\eps$-average agnostic learner. We obtain a slightly weaker lower bound that is essentially $T\geq \frac{d}{62\eps^2}$:

\begin{theorem}\label{thm:infoclassicalagnoptconstant}
Let $\C$ be a concept class with VC-dim$(\C)=d$. Then for every $\eps \in (0,1/10]$, there exists a distribution for which every $\varepsilon$-average agnostic learner has sample complexity at least 
$\frac{d}{\eps^2}\cdot \Big( \frac{1}{62}-\frac{\log(2d+2)}{4d}\Big)$.
\end{theorem}

\begin{proof}
The proof is similar to Theorem~\ref{thm:infoclassicalagn}. Assume an $\eps$-average agnostic learner for $\C$ that uses $T$ samples.
For shattered set $\Sh=\{s_1,\ldots,s_d\} \subseteq \01^n$ and $a\in\01^d$, define distribution $D_a$ on $[d]\times\01$ by $D_a(i,\ell)=(1+(-1)^{a_i+\ell}\beta\eps)/2d$, for some constant $\beta\geq 2$ which we shall pick later. 

Again let random variable $\mathbf{A}\in\01^d$ be uniformly random, corresponding to the values of concept $c_a$ on $\Sh$, and $\mathbf{B}=\mathbf{B}_1\ldots\mathbf{B}_T$ be $T$ i.i.d.\ samples from~$D_a$. Note that $c_a$ is the minimal-error concept from $\C$ w.r.t.\ $D_a$, and concept $c_{\tilde{a}}$ has additional error $d_H(a,\tilde{a})\cdot \beta\eps/d$.
Accordingly, an $\eps$-average agnostic learner has to produce (from $\mathbf{B}$) an $h(\mathbf{B})\in\01^d$, which satisfies~$\Ex_{\mathbf{A},\mathbf{B}} [d_H(\mathbf{A},h(\mathbf{B}))]\leq~d/\beta$.

Our three-step analysis is very similar to Theorem~\ref{thm:infoclassicalagn};
only the first step changes:
\begin{enumerate}
\item $I(\mathbf{A}:\mathbf{B})\geq d(1-H(1/\beta))-\log (d+1)$. \\[1mm]
\emph{Proof.} Define random variable $\mathbf{Z}=d_H(\mathbf{A},h(\mathbf{B}))$,
then $\Ex[\mathbf{Z}]\leq d/\beta$. Note that given a string $h(\mathbf{B})$ that is $\ell$-close to $\mathbf{A}$, $\mathbf{A}$ ranges over a set of only ${d\choose \ell}\leq 2^{H(\ell/d)d}$ possible $d$-bit strings (using Fact~\ref{fact:binomialupperbound}), hence $H(\mathbf{A}\mid\mathbf{B},\mathbf{Z}=\ell)\leq H(\mathbf{A}\mid h(\mathbf{B}),\mathbf{Z}=\ell)\leq H(\ell/d)d$. We now lower bound $I(\mathbf{A}:\mathbf{B})$
\begin{align*}
I(\mathbf{A}:\mathbf{B}) & = H(\mathbf{A})-H(\mathbf{A}\mid\mathbf{B})\\
       &  \geq H(\mathbf{A})-H(\mathbf{A}\mid \mathbf{B},\mathbf{Z})-H(\mathbf{Z})\\ 
       & = d-\sum_{\ell=0}^{d+1} \Pr[\mathbf{Z}=\ell]\cdot H(\mathbf{A}\mid \mathbf{B},\mathbf{Z}=\ell)- H(\mathbf{Z})\\
       & \geq d-\Ex_{\ell\in \{0,\ldots,d\}} [H(\ell/d)d]-\log (d+1) \tag{since $\mathbf{Z}\in \{0,\ldots,d\}$}\\
       & \geq d-d H\Big(\frac{\Ex_{\ell} [\ell]}{d}\Big) -\log (d+1)\tag{using Jensen's inequality}\\
       & \geq d-dH(1/\beta)-\log (d+1) \tag{using $\Ex[\mathbf{Z}]\leq d/\beta$},
\end{align*}
where for the third inequality we used the concavity of the binary entropy function to conclude $\Ex_\ell [H(\ell/d)]\leq H(\Ex_\ell [\ell]/d)$, and for the fourth inequality we used that $\beta\geq 2$. 
\item $I(\mathbf{A}:\mathbf{B})\leq T\cdot I(\mathbf{A}:\mathbf{B}_1)$.
\item $I(\mathbf{A}:\mathbf{B}_1)=1-H(1/2+\beta\eps/2)\stackrel{\text{Fact}~\ref{fact:taylorseriesbinaryentropy}}{\leq} \beta^2\eps^2/\ln 4+O(\eps^4)$.
\end{enumerate}
Combining these three steps implies  
$$
T\geq \frac{d\ln 4}{\eps^2} \cdot  \Big( \frac{1-H(1/\beta)}{\beta^2+O(\eps^2)}-\frac{\log(d+1)}{\beta^2d+O(d\eps^2)}\Big).
$$
Using $\eps\leq 1/10$, $\beta=4$ to optimize this lower bound, we obtain~$T\geq \frac{d}{\eps^2}\cdot \Big( \frac{1}{62}-\frac{\log(2d+2)}{4d}\Big)$.
\end{proof}

\subsection{Information-theoretic lower bounds on sample complexity: quantum case}
\label{section:infotheoreticlowerboundsquantum}

Here we will ``quantize'' the above two classical information-theoretic proofs, yielding lower bounds for quantum sample complexity (in both the PAC and the agnostic setting) that are tight up to a logarithmic factor.

\subsubsection{Near-optimal lower bound for quantum PAC learning}

\begin{theorem}\label{thm:infoquantumpac}
Let $\C$ be a concept class with VC-dim$(\C)=d+1$. Then, for every $\delta\in (0,1/2)$ and $\eps \in (0,1/4)$, every $(\eps,\delta)$-PAC quantum learner for $\C$ has sample complexity~$\Omega\Big(\frac{d}{\eps\log(d/\eps)} + \frac{\log(1/\delta)}{\eps}\Big)$.
\end{theorem}

\begin{proof}
The proof is analogous to Theorem~\ref{thm:infoclassicalpac}.
We use the same distribution $D$, with the $\mathbf{B}_i$ now being quantum samples: $\ket{\psi_a}=\sum_{i\in\{0,1,\ldots,d\}}\sqrt{D(s_i)}\ket{i,c_a(s_i)}$. 
The $\mathbf{A}\mathbf{B}$-system is now in the following classical-quantum state:
$$
\frac{1}{2^d}\sum_{a\in\01^d} \ketbra{a}{a}\otimes \ketbra{\psi_a}{\psi_a}^{\otimes T}.
$$
The first two steps of our argument are identical to Theorem~\ref{thm:infoclassicalpac}. We only need to re-analyze step~3:
\begin{enumerate}
\item $I(\mathbf{A}:\mathbf{B})\geq (1-\delta)(1-H(1/4))d - H(\delta)=\Omega(d)$.
\item $I(\mathbf{A}:\mathbf{B})\leq T\cdot I(\mathbf{A}:\mathbf{B}_1)$.
\item $I(\mathbf{A}:\mathbf{B}_1)\leq H(4\eps) + 4\eps\log(2d)=O(\eps\log(d/\eps))$.\\[1mm]
\emph{Proof.} Since $\mathbf{A}\mathbf{B}$ is a classical-quantum state, we have 
$$
I(\mathbf{A}:\mathbf{B}_1)= S(\mathbf{A})+S(\mathbf{B}_1)-S(\mathbf{A}\mathbf{B}_1)=S(\mathbf{B}_1),
$$ 
where the first equality follows from definition and the second equality uses $S(\mathbf{A})=d$ since $\mathbf{A}$ is uniformly distributed in $\01^d$, and $S(\mathbf{A}\mathbf{B}_1)=d$ since the matrix $\sigma=\frac{1}{2^d} \sum_{a\in \01^d} \ketbra{a}{a}\otimes \ketbra{\psi_a}{\psi_a}$ is block diagonal with $2^d$ rank-1 blocks on the diagonal. It thus suffices to bound the entropy of the singular values of the reduced state of $\mathbf{B}_1$, which~is
$$
\rho=\frac{1}{2^d}\sum_{a\in\01^d}\ketbra{\psi_a}{\psi_a}.
$$
Let $\sigma_0\geq \sigma_1\geq\cdots\geq \sigma_{2d}\geq 0$ be its singular values. Since~$\rho$ is a density matrix, these form a probability distribution. Note that the upper-left entry of the matrix $\ketbra{\psi_a}{\psi_a}$ is $D(s_0)=1-4\eps$, hence so is the upper-left entry of $\rho$. This implies $\sigma_0\geq 1-4\eps$. Consider sampling a number $\mathbf{N}\in\{0,1,\ldots,2d\}$ according to the $\sigma$-distribution. Let $\mathbf{Z}$ be the indicator random variable for the event $\mathbf{N}\neq 0$, which has probability~$1-\sigma_0\leq 4\eps$. Note that $H(\mathbf{N}\mid\mathbf{Z}=0)=0$, because $\mathbf{Z}=0$ implies $\mathbf{N}=0$. Also, $H(\mathbf{N}\mid\mathbf{Z}=1)\leq\log(2d)$, because if $\mathbf{Z}=1$ then $\mathbf{N}$ ranges over $2d$ elements.
We now have 
\begin{align*}
S(\rho) & =H(\mathbf{N})=H(\mathbf{N},\mathbf{Z})=H(\mathbf{Z})+H(\mathbf{N}\mid\mathbf{Z})\\
& = H(\mathbf{Z}) + \Pr[\mathbf{Z}=0]\cdot H(\mathbf{N}\mid\mathbf{Z}=0) + \Pr[\mathbf{Z}=1]\cdot H(\mathbf{N}\mid\mathbf{Z}=1)\\
& \leq H(4\eps) + 4\eps\log(2d)\\
& = O(\eps\log(d/\eps)) \tag{using Fact~\ref{fact:taylorseriesbinaryentropy}}.
\end{align*}
\end{enumerate}
Combining these three steps implies $T=\Omega\Big(\frac{d}{\eps\log(d/\eps)}\Big)$.
\end{proof}

\subsubsection{Near-optimal lower bound for quantum agnostic learning}

\begin{theorem}\label{thm:infoquantumagn}
Let $\C$ be a concept class with VC-dim$(\C)=d$. Then for every $\delta\in (0,1/2)$ and $\eps \in (0,1/4)$, every $(\eps,\delta)$-agnostic quantum learner for~$\C$ has sample complexity 
$\Omega\Big(\frac{d}{\eps^2\log(d/\eps)}+\frac{\log(1/\delta)}{\eps^2}\Big)$.
\end{theorem}

\begin{proof}
The proof is analogous to Theorem~\ref{thm:infoclassicalagn},
with the $\mathbf{B}_i$ now being quantum samples for~$D_a$, 
$\ket{\psi_a}=\sum_{i\in[d],\ell\in\01}\sqrt{D_a(i,\ell)}\ket{i,\ell}$.
Again we only need to re-analyze step~3:
\begin{enumerate}
\item $I(\mathbf{A}:\mathbf{B})\geq (1-\delta)(1-H(1/4))d - H(\delta)=\Omega(d)$.
\item $I(\mathbf{A}:\mathbf{B})\leq T\cdot I(\mathbf{A}:\mathbf{B}_1)$.
\item $I(\mathbf{A}:\mathbf{B}_1)=O(\eps^2\log(d/\eps))$.\\[1mm]
\emph{Proof (of step~3).} As in step~3 of the proof of Theorem~\ref{thm:infoquantumpac}, it suffices to upper bound the entropy of
$$
\rho=\frac{1}{2^d}\sum_{a\in\01^d}\ketbra{\psi_a}{\psi_a}.
$$
We now lower bound the largest singular value of~$\rho$.  Consider $\ket{\psi}=\frac{1}{\sqrt{2d}}\sum_{i\in[d],\ell\in\01}\ket{i,\ell}$.
$$
\ip{\psi}{\psi_a}=\frac{1}{d}\sum_{i\in[d]}\frac{1}{2}\Big(\sqrt{1+4\eps}+\sqrt{1-4\eps}\Big)=\frac{1}{2}\Big(\sqrt{1+4\eps}+\sqrt{1-4\eps}\Big)\geq 1-2\eps^2-O(\eps^4),
$$
where the last inequality used the Taylor series expansion of $\sqrt{1+x}$.
This implies that the largest singular value of $\rho$ is at least
$$
\bra{\psi}\rho\ket{\psi}=\frac{1}{2^d}\sum_{a\in\01^d}|\ip{\psi}{\psi_a}|^2\geq 1-4\eps^2-O(\eps^4).
$$
We can now finish as in step~3 of the proof of Theorem~\ref{thm:infoquantumpac}:
\begin{align*}
I(\mathbf{A}:\mathbf{B}_1)\leq S(\rho)\leq H(4\eps^2) + 4\eps^2\log(2d)\stackrel{\text{Fact}~\ref{fact:taylorseriesbinaryentropy}}{=}O(\eps^2\log(d/\eps)).
\end{align*}
\end{enumerate}
Combining these three steps implies $T=\Omega\Big(\frac{d}{\eps^2\log(d/\eps)}\Big)$.
\end{proof}

\section{A lower bound by analysis of state identification}
\label{section:stateident}

In this section we present a tight lower bound on quantum sample complexity for both the PAC and the agnostic learning models, using ideas from Fourier analysis to analyze the performance of the Pretty Good Measurement. The core of both lower bounds is the following combinatorial~theorem.

\begin{theorem}
\label{thm:upperboundonsqrtGram}
For $m\geq 10$, let $f:\01^m\rightarrow \R$ be defined as $f(z)=(1-\beta\frac{|z|}{m})^T$ for some $\beta\in (0,1]$ and $T\in [1, m/(e^3\beta)]$. For $k\leq m$, let $M\in \F_2^{m\times k}$ be a matrix with rank $k$. Suppose $A\in \R^{2^{k}\times 2^{k}}$ is defined as $A(x,y)=(f\circ M)(x+ y)$ for $x,y\in \01^k$,~then 
$$
 \sqrt{A}(x,x)\leq  \frac{2\sqrt{e}}{2^{k/2}}\Big(1-\frac{\beta}{2}\Big)^{T/2} e^{11T^2\beta^2/m+\sqrt{Tm\beta}}    \qquad \text{for all } x\in \01^k.
$$
\end{theorem}

\begin{proof}
The structure of the proof is to first diagonalize $A$, relating its eigenvalues to the Fourier coefficients of~$f$. This allows to calculate the diagonal entries of $\sqrt{A}$ exactly in terms of those Fourier coefficients. We then upper bound those Fourier coefficients using a combinatorial argument.

We first observe the well-known relation between the eigenvalues of a matrix $P$ defined as $P(x,y)=g(x+y)$ for $x,y\in \01^k$, and the Fourier coefficients of $g$.
\begin{claim}
\label{claim:symmmatrixeigenvaluesfourier}
Suppose $g:\01^k\rightarrow \R$ and $P\in \R^{2^k\times 2^k}$ is defined as $P(x,y)=g(x+y)$, then the eigenvalues of $P$ are $\{2^k\widehat{g}(Q):Q\in \01^k\}$.
\end{claim}

\begin{proof}
Let $H \in \R^{2^k\times 2^k}$ be the matrix defined as $H(x,y)=(-1)^{x\cdot y}$ for $x,y \in \01^k$. It is easy to see that $H^{-1}(x,y)=(-1)^{x\cdot y}/2^k$. We now show that $H$ diagonalizes~$P$:
\begin{align*}
(HPH^{-1})(x,y)&=\frac{1}{2^k}\sum_{z_1,z_2 \in \01^{k}}(-1)^{z_1\cdot x+ z_2\cdot y}g(z_1+ z_2) \\
&=\frac{1}{2^k}\sum_{z_1,z_2,Q \in \01^k}(-1)^{z_1\cdot x+ z_2\cdot y}\widehat{g}(Q) (-1)^{Q\cdot(z_1+ z_2)} \tag{Fourier expansion of $g$}\\
&= \frac{1}{2^k}\sum_{Q\in\01^k}\widehat{g}(Q)\sum_{z_1\in\01^k}(-1)^{(x+Q)\cdot z_1}\sum_{z_2\in\01^k}(-1)^{(y+Q)\cdot z_2}\\ 
&= 2^k \widehat{g}(x) \delta_{x,y}\tag{using $\sum_{z\in \01^k} [(-1)^{(a+b)\cdot z}]=2^k\delta_{a,b}$} 
\end{align*}
The eigenvalues of $P$ are the diagonal entries, $\{2^k\widehat{g}(Q): Q\in \01^k\}$.
\end{proof}

We now relate the diagonal entries of $\sqrt{A}$ to the Fourier coefficients of $f$:

\begin{claim}
\label{claim:sqrtGram}
For all $x\in \01^k$, we have 
$$
\sqrt{A}(x,x)=\frac{1}{2^{k/2}}\sum_{Q\in \01^k} \sqrt{ \sum_{\substack{S\in \01^m:M^t S =Q}} \widehat{f} (S)} 
.
$$
\end{claim}

\begin{proof}
Since $A(x,y)=(f\circ M)(x+ y)$, by Claim~\ref{claim:symmmatrixeigenvaluesfourier} it follows that $H$ (as defined in the proof of Claim~\ref{claim:symmmatrixeigenvaluesfourier}) diagonalizes $A$ and the eigenvalues of~$A$ are $\{2^k\widehat{f\circ M} (Q):Q\in \01^k\}$. Hence, we have 
$$
\sqrt{A}=H^{-1}\cdot \text{diag}\Big(\Big\{\sqrt{2^k\widehat{f\circ M} (Q)}:Q\in \01^k\Big\}\Big)\cdot H,
$$ 
and the diagonal entries of $\sqrt{A}$~are
\begin{align*}
\sqrt{A}(x,x)&=\frac{1}{2^{k/2}}\sum_{Q\in \01^{k}} \sqrt{ \widehat{f\circ M} (Q)} \stackrel{\text{Claim}~\ref{claim:Fouriercoeffmatrixprod}}{=}\frac{1}{2^{k/2}}\sum_{Q\in \01^{k}} \sqrt{ \sum_{S\in \01^m:M^t S =Q} \widehat{f} (S)}.
\end{align*}
\end{proof}

In the following lemma, we give an upper bound on the Fourier coefficients of $f$, which in turn (from the claim above) gives an upper bound on the diagonal entries of $\sqrt{A}$.

\begin{lemma}
\label{lemma:Fouriercoefficientupperbound}
For $\beta\in (0,1]$, the Fourier coefficients of $f:\01^m\rightarrow \R$ defined as $f(z)=(1-\beta\frac{|z|}{m} )^T$, satisfy 
\begin{align*}
0 \leq  \widehat{f}(S) \leq 4e\Big(1-\frac{\beta}{2}\Big)^T \Big( \frac{T\beta}{m} \Big)^{q}e^{22T^2\beta^2/m}, \quad \text{for all S such that } |S|=q.
\end{align*}
\end{lemma}

\begin{proof}
In order to see why the Fourier coefficients of $f$ are non-negative, we first define the set $U=\{u_x^{\otimes T}\}_{x\in \01^m}$ where $u_x=\sqrt{1-\beta}\ket{0,0}+\sqrt{\beta/m}\sum_{i\in [m]}\ket{i,x_i}$. Let $V$ be the $2^m\times 2^m$ Gram matrix for the set~$U$. For $x,y\in \01^m$, we have
\begin{align*}
V(x,y)=(u_x^* u_y)^T&=\Big(1-\beta+\frac{\beta}{m}\sum_{i=1}^m \ip {x_i}{y_i}\Big)^T\\
&=\Big(1-\beta+\frac{\beta}{m}(m-|x+ y|)\Big)^T\\
&=\Big(1-\beta\frac{|x+ y|}{m} \Big)^T =f(x+ y).
\end{align*}
By Claim~\ref{claim:symmmatrixeigenvaluesfourier}, the eigenvalues of the Gram matrix $V$ are $\{2^m\widehat{f}(S):S\in \01^m\}$. Since the Gram matrix is psd, its eigenvalues are non-negative, which implies that $\widehat{f}(S)\geq 0$ for all $S\in \01^m$.

We now prove the upper bound in the lemma. By definition,
\begin{align*}
\widehat{f}(S)&=\Ex_{z\in \01^m} \Big[\Big(1-\beta\frac{|z|}{m}\Big)^T (-1)^{S\cdot z}\Big] \\
&=\Ex_{z\in \01^m} \Big[\Big(1-\frac{\beta}{2}+\frac{\beta}{2m} \sum_{i=1}^m (-1)^{z_i}\Big)^T (-1)^{S\cdot z} \Big]\tag{since $|z|=\sum_{i\in [m]} \frac{1-(-1)^{z_i}}{2}$}\\
&=\sum_{\ell=0}^T \binom{T}{\ell} \Big(1-\frac{\beta}{2}\Big)^{T-\ell}\Big(\frac{\beta}{2m}\Big)^{\ell} \Ex_{z\in \01^m} \Big[\sum_{i_1,\ldots, i_\ell=1}^m (-1)^{z\cdot (e_{i_1}+\cdots+e_{i_\ell}+S)}\Big]\\
&=\sum_{\ell=0}^T \binom{T}{\ell} \Big(1-\frac{\beta}{2}\Big)^{T-\ell}\Big(\frac{\beta}{2m}\Big)^{\ell}\sum_{i_1,\ldots, i_\ell=1}^m 1_{[e_{i_1}+\cdots+e_{i_\ell}=S]} \tag{using $\Ex_{z\in \01^m} [(-1)^{(z_1+z_2)\cdot z}]=\delta_{z_1,z_2}$}
\end{align*}
We will use the following claim to upper bound the combinatorial sum in the quantity~above.
\begin{claim}
\label{claim:combinatorialcounting}
Fix $S\in \01^m$ with Hamming weight $|S|=q$. For every $\ell\in \{q,\ldots, T\}$, we have
\[\sum_{i_1,\ldots, i_\ell=1}^m 1_{[e_{i_1}+\cdots+e_{i_\ell}=S]}\leq 
 \begin{cases} 
       \ell! \cdot m^{(\ell-q)/2}\Big/ \Big(2^{(\ell-q)/2} ((\ell-q)/2)!\Big) &  \text{if }(\ell-q) \text{ is even } \\
      0 & \text{otherwise} 
   \end{cases}
\]
\end{claim}
\begin{proof}
Since $|S|=q$, we can write $S=e_{r_1}+\cdots +e_{r_q}$ for distinct $r_1,\ldots,r_q \in [m]$. There are $\binom {\ell}{q}$ ways to pick $q$ indices in $(i_1,\ldots ,i_\ell)$ (w.l.o.g.\ let them be $i_1,\ldots,i_q$) and there are $q!$ factorial ways to assign $(r_1,\ldots,r_q)$ to $(i_1,\ldots,i_q)$. It remains to count the number of ways that we can assign values to the remaining indices $i_{q+1},\ldots, i_\ell$ such that $e_{i_{q+1}}+\cdots +e_{i_\ell} =0$. If $\ell-q$ is odd then this number is~0, so from now on assume $\ell-q$ is even. We upper bound the number of such assignments by partitioning the $\ell-q$ indices into pairs and assigning the same value to both indices in each pair.

We first count the number of ways to partition a set of $\ell-q$ indices into subsets of size $2$. This number is exactly $(\ell-q)! \Big(2^{(\ell-q)/2}((\ell-q)/2)!\Big)^{-1}$. Furthermore, there are $m$ possible values that can be assigned to the pair of indices in each of the $(\ell-q)/2$ subsets such that $e_i+e_j=0$ within each subset. Note that assigning $m$ possible values to each pair of indices in the $(\ell-q)/2$ subsets overcounts, but this rough upper bound is sufficient for our purposes. 

Combining the three arguments, we conclude 
  $$
  \sum_{i_1,\ldots, i_\ell=1}^d 1_{[e_{i_1}+\cdots+e_{i_\ell}=S]}\leq  \binom{\ell}{q} q! \cdot (\ell-q)!\cdot m^{(\ell-q)/2}\Big/\Big(2^{(\ell-q)/2} ((\ell-q)/2)!\Big).
  $$
which yields the claim.
\end{proof}

Continuing with the evaluation of the Fourier coefficient and using the claim above, we~have
\begin{align*}
\widehat{f}(S)&=\sum_{\ell=0}^T \binom{T}{\ell} \Big(1-\frac{\beta}{2}\Big)^{T-\ell}\Big(\frac{\beta}{2m}\Big)^{\ell}\sum_{i_1,\ldots, i_\ell=1}^m 1_{[e_{i_1}+\cdots+e_{i_\ell}=S]} \\
&\leq \sum_{\ell=q}^T \binom{T}{\ell} \Big(1-\frac{\beta}{2}\Big)^{T-\ell}\Big(\frac{\beta}{2m}\Big)^{\ell}\ell! \cdot m^{(\ell-q)/2}\Big/\Big(2^{(\ell-q)/2} \Big(\frac{\ell-q}{2}\Big)!\Big) \tag{by Claim~\ref{claim:combinatorialcounting}}\\
&= \Big(1-\frac{\beta}{2}\Big)^T \Big( \frac{2}{m} \Big)^{q/2}\sum_{\ell=q}^T \binom{T}{\ell} \ell! \Big(\frac{\beta}{m(2-\beta)}\Big)^\ell \Big( \frac{m}{2} \Big)^{\ell/2} \Big \slash \Big(\frac{\ell-q}{2}\Big)!\\
& \leq \Big(1-\frac{\beta}{2}\Big)^T \Big( \frac{2}{m} \Big)^{q/2}\sum_{\ell=q}^T \Big(T \cdot \frac{\beta}{m}  \cdot \sqrt{\frac{m}{2}}  \Big)^\ell\Big \slash \Big(\frac{\ell-q}{2}\Big)!\tag{since $\beta< 1$ and $\binom{T}{\ell} \ell!\leq T^\ell$}\\
&=\Big(1-\frac{\beta}{2}\Big)^T \Big( \frac{T\beta}{m} \Big)^{q}\sum_{r=0}^{T-q} \Big( \frac{T\beta}{\sqrt{2m}} \Big)^{r} \frac{1}{(r/2)!} \tag{substituting $r\leftarrow (\ell-q)$}\\
&\leq \Big(1-\frac{\beta}{2}\Big)^T \Big( \frac{T\beta}{m} \Big)^{q}\sum_{r=0}^{T-q} \Big( \frac{T\beta}{\sqrt{2m}} \Big)^{r} \frac{e^{r/2}}{ (r/2)^{r/2}} \tag{using $n!\geq (n/e)^n$}\\
&=  \Big(1-\frac{\beta}{2}\Big)^T \Big( \frac{T\beta}{m} \Big)^{q}\sum_{r=0}^{T-q} \Big( \frac{\sqrt{e}T\beta}{\sqrt{mr}} \Big)^{r}\\
&\leq \Big(1-\frac{\beta}{2}\Big)^T \Big( \frac{T\beta}{m} \Big)^{q}\sum_{r=0}^{T} \Big( \frac{\sqrt{e}T\beta}{\sqrt{mr}} \Big)^{r} \tag{since the summands are $\geq 0$} \\
&= \Big(1-\frac{\beta}{2}\Big)^T \Big( \frac{T\beta}{m} \Big)^{q}\Big(\sum_{r=0}^{\lceil e^3T^2\beta^2/m \rceil} \Big( \frac{\sqrt{e}T\beta}{\sqrt{mr}} \Big)^{r}+\sum_{r=\lceil e^3T^2\beta^2/m \rceil+1}^{T} \Big( \frac{\sqrt{e}T\beta}{\sqrt{mr}} \Big)^{r}\Big).
\end{align*}
Note that by the assumptions of the theorem, $T^2e^3\beta^2/m\leq T\beta \leq T$, which allowed us to split the sum into two pieces in the last equality.
At this point, we upper bound both pieces in the last equation separately. For the first piece, using Claim~\ref{claim:maximizingc/sqrtt} it follows that $\Big( \frac{\sqrt{e}T\beta}{\sqrt{m r}} \Big)^{r}$ is maximized at $r=\lceil T^2\beta^2/m \rceil$. Hence we get
\begin{align}
\label{eq:upperboundfirstterm}
\sum_{r=0}^{\lceil e^3T^2\beta^2/m\rceil} \Big( \frac{\sqrt{e}T\beta}{\sqrt{mr}} \Big)^{r}\leq \Big(2+\frac{e^3T^2\beta^2}{m} \Big)e^{\lceil T^2\beta^2/m \rceil/2}\leq 2e^{22T^2\beta^2/m+1}, 
\end{align}
where the first inequality uses Claim~\ref{claim:maximizingc/sqrtt} and the second inequality uses $2+x\leq 2e^x$ for $x\geq 0$ and $e^3+1/2\leq 22$.
 For the second piece, we  use 
\begin{align}
\label{eq:upperboundsecondterm}
\sum_{r={\lceil e^3T^2\beta^2/m\rceil+1}}^{T} \Big( \frac{\sqrt{e}T\beta}{\sqrt{mr}} \Big)^{r} \leq \sum_{r={\lceil e^3T^2\beta^2/m\rceil+1}}^{T}\Big( \frac{1}{e} \Big)^{r} \leq  \sum_{r=1}^{T} \Big( \frac{1}{e} \Big)^{r} = \frac{1-e^{-T}}{e-1} \leq 2/3. 
\end{align}
So we finally get
\begin{align*}
\widehat{f}(S)&\leq \Big(1-\frac{\beta}{2}\Big)^T \Big( \frac{T\beta}{m} \Big)^{q}\Big(2e^{22T^2\beta^2/m+1}+2/3\Big) \tag{using Eq.~(\ref{eq:upperboundfirstterm}), (\ref{eq:upperboundsecondterm})}\\
&\leq 4e\Big(1-\frac{\beta}{2}\Big)^T \Big( \frac{T\beta}{m} \Big)^{q}e^{22T^2\beta^2/m} \tag{since $22T^2\beta^2/m>0$}
\end{align*}
\end{proof}
The theorem follows by putting together Claim~\ref{claim:sqrtGram} and Lemma~\ref{lemma:Fouriercoefficientupperbound}:
\begin{align*}
\sqrt{A}(x,x)&=\frac{1}{2^{k/2}}\sum_{Q\in \01^{k}} \sqrt{ \sum_{S\in \01^m:M^t S =Q} \widehat{f} (S)}\tag{using Claim~\ref{claim:sqrtGram}}\\ 
&\leq \frac{1}{2^{k/2}}\sum_{Q\in \01^{k}} \sum_{S\in \01^m:M^t S =Q} \sqrt{ \widehat{f} (S)} \tag{using lower bound from Lemma~\ref{lemma:Fouriercoefficientupperbound}}\\ 
&= \frac{1}{2^{k/2}} \sum_{S\in \01^m} \sqrt{ \widehat{f} (S)} \tag{$\cup_Q\{S:M^t S=Q\}=\01^m$ since rank($M$)=$k$} \\
&= \frac{1}{2^{k/2}} \sum_{q=0}^m\sum_{S\in \01^m:|S|=q} \sqrt{ \widehat{f} (S)} \\
&\leq \frac{2\sqrt{e}}{2^{k/2}}\Big(1-\frac{\beta}{2}\Big)^{T/2} e^{11T^2\beta^2/m} \sum_{q=0}^m \binom{m}{q}  \Big( \frac{T\beta}{m} \Big)^{q/2} \tag{using Lemma~\ref{lemma:Fouriercoefficientupperbound}} \\
&= \frac{2\sqrt{e}}{2^{k/2}}\Big(1-\frac{\beta}{2}\Big)^{T/2} e^{11T^2\beta^2/m}\Big(1+\sqrt{\frac{T\beta}{m} }\Big)^m  \tag{using binomial theorem}\\
&\leq  \frac{2\sqrt{e}}{2^{k/2}}\Big(1-\frac{\beta}{2}\Big)^{T/2} e^{11T^2\beta^2/m+\sqrt{Tm\beta}} \tag{using $(1+x)^t\leq e^{xt}$ for~$x,t\geq 0$}.
\end{align*}
\end{proof}

\subsection{Optimal lower bound for quantum PAC learning}
\label{section:optimalpaclowerbounds}

We can now prove our tight lower bound on quantum sample complexity in the PAC model:

\begin{theorem}
\label{thm:optimalpaclowerbound}
Let $\C$ be a concept class with VC-dim$(\C)=d+1$, for sufficiently large~$d$. Then for every $\delta\in (0,1/2)$ and $\eps \in (0,1/20)$, every $(\eps,\delta)$-PAC quantum learner for~$\C$ has sample complexity 
$\Omega\Big(\frac{d}{\eps} + \frac{1}{\eps}\log \frac{1}{\delta}\Big)$.
\end{theorem}

\begin{proof}
The $d$-independent part of the lower bound is Lemma~\ref{lemma:deltapartoflemmapac}. To prove the $d$-dependent part, define a distribution $D$ on a set $\Sh=\{s_0,\ldots,s_d\} \subseteq \01^n$ that is shattered by $\C$ as follows: $D(s_0)=1-20\eps$ and $D(s_i)=20\eps/d$ for all $i\in [d]$. 

Now consider a $[d,k,r]_2$ linear code (for $k\geq d/4$, distance $r\geq d/8$) as shown to exist in Theorem~\ref{thm:randomlinearcode} with the generator matrix $M\in \F_2^{d\times k}$ of rank~$k$. Let $\{Mx:x\in \01^k\} \subseteq \01^{d}$ be the set of codewords in this linear code; these satisfy $d_H(Mx,My)\geq d/8$ whenever $x\neq y$. For each $x\in \01^k$, let $c^x$ be a concept defined on the shattered set as: $c^x(s_0)=0$ and $c^x(s_i)=(Mx)_i$ for all $i\in [d]$. The existence of such concepts in $\C$ follows from the fact that $\Sh$ is shattered by~$\C$. From the distance property of the code, we have $\Pr_{s\sim D}[c^x(s)\neq c^y(s)]\geq \frac{20\eps}{d} \frac{d}{8}=~5\eps/2$. This in particular implies that an $(\eps,\delta)$-PAC quantum learner that tries to $\eps$-approximate a concept from $\{c^x:x\in\01^k\}$ should successfully \emph{identify} that concept with probability at least~$1-\delta$.  

We now consider the following state identification problem: for $x\in \01^k$, denote $\ket{\psi_x}=\sum_{i\in \{0,\ldots, d\}} \sqrt{D(s_i)} \ket{s_i, c^x(s_i)}$. Let the $(\eps,\delta)$-PAC quantum sample complexity be~$T$. Assume $T\leq d/(20e^3 \eps)$, since otherwise $T\geq \Omega(d/\eps)$ and the theorem follows. Suppose the learner has knowledge of the ensemble $\E=\{(2^{-k},\ket{\psi_x}^{\otimes T}):x\in \01^k\}$, and is given $\ket{\psi_x}^{\otimes T} \in \E$ for a uniformly random~$x$. The learner would like to maximize the average probability of success to identify the given state. For this problem, we prove a lower bound on $T$ using the PGM defined in Section~\ref{section:pgm}. In particular, we show that using the PGM, if a learner successfully identifies the states in $\E$, then $T=\Omega(d/\eps)$. Since the PGM is the optimal measurement\footnote{For $x\in\01^k$, define unitary $U_{c^x}:  \ket{s_i,b}\rightarrow \ket{s_i,b+c^x(s_i)}$ for all $i\in \{0,\ldots,d\}$. The ensemble $\E$ is generated by applying $\{U_{c^x}\}_{x\in \01^k}$ to $\ket{\ph}=\sum_{i\in \{0,\ldots, d\}} \sqrt{D(s_i)} \ket{s_i,0}$. View $c^x=(0,Mx) \in \01^{d+1}$ as a concatenated string where $Mx$ is a codeword of the $[d,k,r]_2$ code. Since the $2^k$ codewords of the $[d,k,r]_2$ code form a linear subspace, $\{U_{c^x}\}_{x\in \01^k}$ is an Abelian group. From the discussion in Section~\ref{section:pgm}, we conclude that the PGM is the optimal measurement for this state identification~problem.} that the learner could have performed, the result follows. The following lemma makes this lower bound rigorous and will conclude the proof of the~theorem. 

\begin{lemma}\label{lem:PGMboundpac}
For every $x\in \01^k$, let $\ket{\psi_x}=\sum_{i\in \{0,\ldots, d\}} \sqrt{D(s_i)} \ket{s_i, c^x(s_i)}$, and $\E=\{(2^{-k},\ket{\psi_x}^{\otimes T}):x\in \01^k\}$.~Then\footnote{We made no attempt to optimize the constants here.}
$$
P^{PGM} (\E)\leq  \frac{4e}{2^{d/4+T\eps}}e^{8800T^2\eps^2/d+4\sqrt{5Td\eps}}.
$$
\end{lemma}

Before we prove the lemma, we first show why it implies the theorem. Since we observed above that $P^{opt}(\E)=P^{PGM}(\E)$, a good learner satisfies $P^{PGM}(\E)=\Omega(1)$ (say for $\delta=1/4$), which in turn implies
$$
\Omega(\max\{d,T\eps\})\leq O(\min\{T^2\eps^2/d,\sqrt{Td\eps }\}).
$$
Note that if $T\eps$ maximizes the left-hand side, then $d\leq T\eps$ and hence $T\geq \Omega(d/\eps)$. The remaining cases are $\Omega(d)\leq T^2\eps^2/d$ and $\Omega(d)\leq \sqrt{Td\eps}$. Both these statements give us $T\geq \Omega(d/\eps)$. Hence the theorem follows, and it remains to prove Lemma~\ref{lem:PGMboundpac}:
\begin{proof}
Let $\E'=\{2^{-k/2}\ket{\psi_x}^{\otimes T}:x\in \01^k\}$ and $G$ be the $2^k\times 2^k$ Gram matrix for $\E'$.  As we saw in Section~\ref{section:pgm}, the success probability of identifying the states in the ensemble $\E$ using the PGM is
$$
P^{PGM}(\E)=\sum_{x\in \01^k} \sqrt{G}({x,x})^2.
$$
For all $x,y\in \01^k$, the entries of the Gram matrix $G$ can be written as:
\begin{align*}
G(x,y)=\frac{1}{2^k}\ip{\psi_x}{\psi_y}^T&=\frac{1}{2^k}\Big( (1-20\eps)+\frac{20\eps}{d} \sum_{i=1}^d \ip{c^x(s_i)}{c^y(s_i)}\Big)^T \\
&=\frac{1}{2^k} \Big((1-20\eps)+\frac{20\eps}{d} (d-d_H(Mx,My))\Big)^T\\
&=\frac{1}{2^k}\Big(1-\frac{20\eps }{d}d_H(Mx,My)\Big)^T,
\end{align*}
where $Mx$, $My \in \01^{d}$ are codewords  in the linear code defined earlier. Define $f:\01^{d}\rightarrow \R$ as $f(z)=(1-\frac{20\eps}{d}|z|)^T$, and let $A(x,y)=(f\circ M)(x+ y)$ for $x,y\in \01^k$. Note that $G=A/2^k$. Since we assumed $T\leq d/(20e^3 \eps)$, we can use Theorem~\ref{thm:upperboundonsqrtGram} (by choosing $m=d$ and $\beta=20\eps$) to upper bound the success probability of successfully identifying the states in the ensemble $\E$ using the PGM. 
\begin{align*}
P^{PGM}(\E)&=\sum_{x\in \01^k} \sqrt{G}(x,x)^2 \\
&=\frac{1}{2^k}\sum_{x\in \01^k} \sqrt{A}(x,x)^2 \tag{since $G=A/2^k$}\\
&\leq  \frac{4e}{2^{k}}\Big(1-\frac{\beta}{2}\Big)^{T} e^{22T^2\beta^2/d+2\sqrt{Td\beta}}  \tag{using Theorem~\ref{thm:upperboundonsqrtGram}} \\
&=  \frac{4e}{2^{k}}\Big(1-10\eps\Big)^{T} e^{8800T^2\eps^2/d+4\sqrt{5Td\eps}}     \tag{substituting $\beta=20\eps$}\\
&\leq  \frac{4e}{2^{k+T\eps}}e^{8800T^2\eps^2/d+4\sqrt{5Td\eps}}   \tag{using $(1-10\eps)^T\leq e^{-10\eps T}\leq 2^{-\eps T}$}\\
\end{align*}
The lemma follows by observing that $k\geq d/4$.
\end{proof}
\end{proof}

\subsection{Optimal lower bound for quantum agnostic learning}

We now use the same approach to obtain a tight lower bound on quantum sample complexity in the \emph{agnostic} setting.

\begin{theorem}
Let $\C$ be a concept class with VC-dim$(\C)=d$, for sufficiently large $d$. Then for every $\delta\in (0,1/2)$ and $\eps \in(0,1/10)$, every $(\eps,\delta)$-agnostic quantum learner for~$\C$ has sample complexity 
$\Omega\Big(\frac{d}{\eps^2} + \frac{1}{\eps^2}\log \frac{1}{\delta}\Big)$.
\end{theorem}

\begin{proof}
The $d$-independent part of the lower bound is Lemma~\ref{lemma:deltapartoflemmaagnostic}. For the $d$-dependent term in the lower bound, consider a $[d,k,r]_2$ linear code (for $k\geq d/4$, distance $r\geq d/8$) as shown to exist in Theorem~\ref{thm:randomlinearcode}, with generator matrix $M\in \F_2^{d\times k}$ of rank~$k$. Let $\{Mx:x\in \01^k\} \subseteq \01^{d}$ be the set of $2^k$ codewords in this linear code; these satisfy $d_H(Mx,My)\geq d/8$ whenever $x\neq y$. To each codeword $x\in\01^k$ we associate a distribution $D_x$ as follows:
$$
D_x(s_i,b)=\frac{1}{d} \Big (\frac{1}{2}+\frac{1}{2}(-1)^{(Mx)_i+b}\alpha \Big), \qquad  \text{for }  (i,b)\in [d]\times \01,
$$ 
where $\Sh=\{s_1,\ldots,s_d\}$ is a set that is shattered by $\C$, and $\alpha$ is a parameter which we shall pick later. Let $c^{x}\in\C$ be a concept that labels $\Sh$ according to $Mx\in \01^d$. The existence of such $c^x\in\C$ follows from the fact that $\Sh$ is shattered by $\C$. Note that $c^{x}$ is the minimal-error concept in $\C$ w.r.t.\ $D_x$. A learner that labels $\Sh$ according to some string $\ell\in\01^d$ has additional error $d_H(Mx,\ell)\cdot \alpha/d$ compared to $c^x$. This in particular implies that an $(\eps,\delta)$-agnostic quantum learner has to find (with probability at least $1-\delta$) an $\ell$ such that $d_H(Mx,\ell)\leq d\eps/\alpha$. We pick $\alpha=20\eps$ and we get $d_H(Mx,\ell)\leq d/20$. However, since $Mx$ was a codeword of a $[d,k,r]_2$ code with distance $r\geq d/8$, finding an $\ell$ satisfying $d_H(Mx,\ell)\leq d/20$ is equivalent to \emph{identifying} $Mx$, and hence~$x$.

Now consider the following state identification problem: let $\ket{\psi_x}=\sum_{(i,b)\in [d]\times \01} \sqrt{D_x(s_i,b)} \ket{s_i, b}$ for $x\in \01^k$. Let the $(\eps,\delta)$-agnostic quantum sample complexity be~$T$. Assume $T\leq d/(100e^3 \eps^2)$, since otherwise $T\geq \Omega(d/\eps^2)$ and the theorem follows.  Suppose the learner has knowledge of the ensemble $\E=\{(2^{-k},\ket{\psi_x}^{\otimes T}):x\in \01^k\}$, and is given $\ket{\psi_{x}}^{\otimes T} \in \E$ for uniformly random~$x$. The learner would like to maximize the average probability of success to identify the given state. For this problem, we prove a lower bound on $T$ using the PGM defined in Section~\ref{section:pgm}. In particular, we show that using the PGM, if a learner successfully identifies the states in $\E$, then $T=\Omega(d/\eps^2)$. Since the PGM is the optimal measurement\footnote{For $x\in\01^k$, define unitary $U_{c^x}=\sum_{i\in [d]} \ketbra{s_i}{s_i}\otimes X^{(Mx)_i}$, where $X$ is the NOT-gate, so $X^{(Mx)_i}\ket{b}=\ket{b+(Mx)_i}$ for $b\in \01$. The ensemble $\E$ is generated by applying $\{U_{c^x}\}_{x\in \01^k}$ to $\ket{\ph}=\frac{1}{\sqrt{d}}\sum_{(i,b)\in [d]\times \01} \sqrt{\frac{1}{2}+\frac{1}{2}(-1)^{b}\alpha}\ket{s_i,b}$. Since the $2^k$ codewords of the $[d,k,r]_2$ code form a linear subspace, $\{U_{c^x}\}_{x\in \01^k}$ is an Abelian group. From the discussion in Section~\ref{section:pgm}, we conclude that the PGM is the optimal measurement for this state identification problem.} that the learner could have performed, the result follows. The following lemma makes this lower bound rigorous and will conclude the proof of the~theorem. 
 
\begin{lemma}\label{lem:PGMboundagn}
For $x\in \01^k$, let  $\ket{\psi_x}=\sum_{(i,b)\in [d]\times \01} \sqrt{D_x(s_i,b)} \ket{s_i, b}$, and $\E=\{(2^{-k},\ket{\psi_x}^{\otimes T}):x\in \01^k\}$. Then
$$
P^{PGM} (\E)\leq \frac{4e}{e^{(d\ln 2)/4+25T\eps^2}} e^{220000T^2\eps^4/d+20\sqrt{Td \eps^2}}.
$$
 \end{lemma}
Before we prove the lemma, we first show why it implies the theorem. Since we observed above that $P^{opt}(\E)=P^{PGM}(\E)$, a good learner satisfies $P^{PGM}(\E)=\Omega(1)$ (say for $\delta=1/4$), which in turn implies
$$
\Omega(\max\{d,T\eps^2\})\leq O(\min\{T^2\eps^4/d,\sqrt{Td\eps^2}\}).
$$
Like in the proof of Theorem~\ref{thm:optimalpaclowerbound}, this implies a lower bound of $T=\Omega(d/\eps^2)$ and proves the theorem. It remains to prove Lemma~\ref{lem:PGMboundagn}:

\begin{proof}
Let $\E'=\{2^{-k/2}\ket{\psi_{x}}^{\otimes T}:x\in \01^k\}$ and $G$ be the $2^k\times 2^k$ Gram matrix for the set $\E'$. As we saw in Section~\ref{section:pgm}, the success probability of identifying the states in the ensemble $\E$ using the PGM is
$$
P^{PGM}(\E)=\sum_{x\in \01^k} \sqrt{G}({x,x})^2.
$$
For all $x,y\in \01^k$, the entries of $G$ can be written as:
\begin{align*}
2^k\cdot G(x,y)&=\ip{\psi_{x}}{\psi_{y}}^T\\
&=\Big(\sum_{(i,b)\in [d]\times \01} \sqrt{D_x(i,b)D_y(i,b)}\Big)^T\\ 	%
&=\Big(\frac{1}{2d}\sum_{(i,b)\in [d]\times \01} \sqrt{(1+10\eps(-1)^{(Mx)_i+b})(1+10\eps(-1)^{(My)_i+b})}\Big)^T \\
&=\Big(\frac{1}{2d}\sum_{\substack{(i,b):\\ (Mx)_i=(My)_i}} (1+10\eps(-1)^{(Mx)_i+b})+\frac{1}{2d}\sum_{\substack{(i,b):\\(Mx)_i\neq (My)_i}} \sqrt{1-100\eps^2}\Big)^T\\
&=\Big(\frac{d-d_H(Mx,My)}{d}+\frac{\sqrt{1-100\eps^2}}{d}d_H(Mx,My)\Big)^T\\
&=\Big(1-\frac{1-\sqrt{1-100\eps^2}}{d}d_H(Mx,My)\Big)^T.
\end{align*}
where we used $\alpha=20\eps$ in the third equality.

Let $\beta=1-\sqrt{1-100\eps^2}$, which is at most $1$ for $\eps\leq 1/10$. 
Define $f:\01^{d}\rightarrow \R$ as $f(z)=(1-\frac{\beta}{d}|z|)^T$, and let $A(x,y)=(f\circ M)(x+y)$ for $x,y\in \01^k$. Then $G=A/2^k$. Note that  $T\leq d/(100e^3\eps^2)\leq d/(e^3\beta)$ (the first inequality is by assumption and the second inequality follows for $\eps\leq 1/10$ and $\beta\leq 1$). Since we assumed $T\leq d/(100e^3 \eps^2)$, we can use Theorem~\ref{thm:upperboundonsqrtGram} (by choosing $m=d$ and $\beta=1-\sqrt{1-100\eps^2}$) to upper bound the success probability of identifying the states in the ensemble $\E$: 
\begin{align*}
P^{PGM}(\E)&=\sum_{x\in \01^k} \sqrt{G}(x,x)^2 \\
&=\frac{1}{2^k}\sum_{x\in \01^k} \sqrt{A}(x,x)^2 \tag{since $G=A/2^k$}\\
&\leq  \frac{4e}{2^{k}}\Big(1-\frac{\beta}{2}\Big)^{T} e^{22T^2\beta^2/d+2\sqrt{Td\beta}}  \tag{using Theorem~\ref{thm:upperboundonsqrtGram}} \\
&\leq   \frac{4e}{2^{k}}\Big(1-\frac{\beta}{2}\Big)^{T} e^{220000T^2\eps^4/d+20\sqrt{Td \eps^2}}     \tag{using $\beta=1-\sqrt{1-100\eps^2}\leq 100\eps^2$}\\
&\leq \frac{4e}{2^k}\Big(1-25\eps^2\Big)^{T} e^{220000T^2\eps^4/d+20\sqrt{Td \eps^2}}    \tag{using $\sqrt{1-100\eps^2} \leq 1-50\eps^2$}\\
&\leq \frac{4e}{e^{k\ln 2+25T\eps^2}} e^{220000T^2\eps^4/d+20\sqrt{Td \eps^2}}   \tag{using $(1-x)^t\leq e^{-xt}$ for $x,t\geq 0$}.
\end{align*}
The lemma follows by observing that $k\geq d/4$.
\end{proof} 
\end{proof}

\subsection{Additional results}

In this section we mention two additional results that can also be obtained using Theorem~\ref{thm:upperboundonsqrtGram}.

\subsubsection{Quantum PAC sample complexity under random classification noise}
In the theorem below, we show a lower bound on the quantum PAC sample complexity under the random classification noise model with noise rate~$\eta$. Recall that in this model, for every $c\in \C$ and distribution $D$, $\eps,\delta > 0$, given access to copies of the $\eta$-noisy state, 
$$
 \sum_{x\in\01^n}\sqrt{(1-\eta)D(x)}\ket{x,c(x)}+\sqrt{\eta D(x)}\ket{x,1-c(x)},
$$
a $(\eps,\delta)$-PAC quantum learner is required to output an hypothesis $h$ such that $\err_D(c,h)\leq \eps$ with probability at least $1-\delta$.

\begin{theorem}
Let $\C$ be a concept class with VC-dim$(\C)=d+1$, for sufficiently large $d$. Then for every $\delta\in (0,1/2)$, $\eps \in (0,1/20)$ and $\eta \in (0,1/2)$, every $(\eps,\delta)$-PAC quantum learner for $\C$ in the PAC setting with random classification noise rate~$\eta$, has sample complexity 
$\Omega\Big(\frac{d}{(1-2\eta)^2\eps} + \frac{\log(1/\delta)}{(1-2\eta)^2\eps}\Big)$.
\end{theorem}

One can use exactly the same proof technique as in Lemma~\ref{lemma:deltapartoflemmapac} and Theorem~\ref{thm:optimalpaclowerbound} to prove this, with only the additional inequality $1-2\sqrt{\eta(1-\eta)}\leq (1-2\eta)^2$, which holds for $\eta\leq 1/2$. We omit the details of the calculation.

\subsubsection{Distinguishing codeword states}
\label{section:distinguisingcodeword}
Ashley Montanaro (personal communication) alerted us to the following interesting special case of our PGM-based result.

Consider an $[n,k,d]_2$ linear code $\{Mx:x\in \01^k\}$, where $M\in \F^{n\times k}_2$ is the rank-$k$ generator matrix of the code, $k=\Omega(n)$, and distinct codewords have Hamming distance at least~$d$.\footnote{Note that throughout this paper $\C$ was a concept class in $\01^n$ and $d$ was the VC dimension of $\C$. The use of $n,d$ in this section has been changed to conform to the convention in coding theory.} For every $x\in \01^k$, define a \emph{codeword state} $\ket{\psi_x}=\frac{1}{\sqrt{n}} \sum_{i\in [n]} \ket{i,(Mx)_i}$. These states form an example of a \emph{quantum fingerprinting} scheme~\cite{bcww:fp}: $2^k$ states whose pairwise inner products are bounded away from~1. How many copies do we need to identify one such~fingerprint?

Let $\E=\{(2^{-k},\ket{\psi_x}):x\in \01^k\}$ be an ensemble of codeword states. Consider the following task: given $T$ copies of an unknown state drawn uniformly from $\E$, we are required to identify the state with probability~$\geq~4/5$. From Holevo's theorem one can easily obtain a lower bound of $T=\Omega(k / \log n)$ copies, since the learner should obtain $\Omega(k)$ bits of information (i.e., identify $k$-bit string~$x$ with probability~$\geq~4/5$), while each copy of the codeword state gives at most $\log n$ bits of information.
 In the theorem below, we improve that $\Omega(k / \log n)$ to the optimal $\Omega(k)$ for constant-rate codes.

\begin{theorem}\label{thm:codewordstates}
Let $\E=\{\ket{\psi_x}=\frac{1}{\sqrt{n}} \sum_{i\in [n]} \ket{i,(Mx)_i}:x\in \01^k\}$, where $M\in \F^{n\times k}_2$ is the generator matrix of an $[n,k,d]_2$ linear code with $k=\Omega(n)$. Then $\Omega(k)$ copies of an unknown state from $\E$ (drawn uniformly at random) are necessary to be able to identify that state with probability at least~$4/5$.
\end{theorem}

One can use exactly the proof technique of Theorem~\ref{thm:optimalpaclowerbound} to prove the theorem. Suppose we are given $T$ copies of the unknown codeword state. Assume $T\leq n$, since otherwise $T\geq n\geq \sqrt{kn}$ and the theorem follows. Observe that the Gram matrix $G$ for $\E'=\{2^{-k/2}\ket{\psi_x}^{\otimes T}:x\in \01^k\}$ can be written as $G(x,y)=\frac{1}{2^k}\Big(1-\frac{|M(x+y)|}{n}\Big)^T$ for $x,y\in \01^k$. Using Theorem~\ref{thm:upperboundonsqrtGram} (choosing $\beta=1$ and $m=n$) to upper bound the success probability of successfully identifying the states in the ensemble $\E$ using the PGM, we obtain
$$
P^{PGM}(\E)\leq  \frac{4e}{2^{k+T}} e^{22T^2/n+2\sqrt{Tn}}  .
$$
As in the proof of Theorem~\ref{thm:optimalpaclowerbound}, this implies the lower bound of Theorem~\ref{thm:codewordstates}. We omit the details of the calculation.

\section{Conclusion}
\label{section:conclusion}
The main result of this paper is that quantum examples give no significant improvement over the usual random examples in passive, distribution-independent settings. Of course, these negative results do not mean that quantum machine learning is useless.  In our introduction we already mentioned improvements from quantum examples for learning under the uniform distribution; improvements from using quantum membership queries; and improvements in time complexity based on quantum algorithms like Grover's and HHL. Quantum machine learning is still in its infancy, and we hope for many more positive results.

We end by identifying a number of open questions for future work:
\begin{itemize}
\item We gave lower bounds on sample complexity for the rather benign random classification noise.  What about other noise models, such a \emph{malicious} noise?
\item What is the quantum sample complexity for learning concepts whose range is $[k]$ rather than $\01$, for some $k>2$? Even the \emph{classical} sample complexity is not fully determined yet~\cite[Section~29.2]{shwartz&david:learningbook}.
\item Classically, it is still an open question whether the $\log(1/\eps)$-factor in the upper bound of\cite{blumer:optimalpacupper} for $(\eps,\delta)$-\emph{proper} PAC learning is necessary. A weaker result (possibly easier to prove) would be to give a $(\eps,\delta)$-\emph{quantum} proper PAC learner without this $\log(1/\eps)$-factor.
\item In the introduction we mentioned a few examples of learning under the \emph{uniform} distribution where quantum examples are significantly more powerful than classical examples.  Can we find more such examples of quantum improvements in sample complexity in fixed-distribution settings?
\item Can we find more examples of quantum speed-up in \emph{time} complexity of learning, for example for learning depth-3 or even constant-depth circuits?
\end{itemize}

\subsubsection*{Acknowledgments.} 
We thank Shalev Ben-David, Dmitry Gavinsky, Robin Kothari, Nishant Mehta, Ashley Montanaro, Henry Yuen for helpful comments and pointers to the literature. We also thank Ashley Montanaro for suggesting the additional remark in Section~\ref{section:distinguisingcodeword}.

\bibliographystyle{alpha}

\begin{thebibliography}{BWP{\etalchar{+}}16}

\bibitem[Aar07]{aaronson:qlearnability}
S.~Aaronson.
\newblock The learnability of quantum states.
\newblock {\em Proceedings of the Royal Society of London}, 463(2088), 2007.
\newblock quant-ph/0608142.

\bibitem[Aar15]{aaronson:fineprint}
S.~Aaronson.
\newblock Quantum machine learning algorithms: Read the fine print.
\newblock {\em Nature Physics}, 11(4):291--293, April 2015.

\bibitem[AB09]{anthony&bartlett:learningbook}
M.~Anthony and P.~L. Bartlett.
\newblock {\em Neural network learning: Theoretical foundations}.
\newblock Cambridge University Press, 2009.

\bibitem[ABG06]{aimeur:mlinquantumworld}
E.~A{\"{\i}}meur, G.~Brassard, and S.~Gambs.
\newblock Machine learning in a quantum world.
\newblock In {\em Proceedings of Advances in Artificial Intelligence, 19th
  Conference of the Canadian Society for Computational Studies of
  Intelligence}, volume 4013, pages 431--442, 2006.

\bibitem[ABG13]{aimeur:qspeedup}
E.~A{\"{\i}}meur, G.~Brassard, and S.~Gambs.
\newblock Quantum speed-up for unsupervised learning.
\newblock {\em Machine Learning}, 90(2):261--287, 2013.

\bibitem[AdW17]{arunachalam:quantumlearningsurvey}
S.~Arunachalam and R.~de~Wolf.
\newblock A survey of quantum learning theory, 2017.
\newblock To appear as Computational Complexity Column in SIGACT News, June
  2017. Preprint at arxiv:1606.08920.

\bibitem[AG98]{apolloni&gentile:act}
B.~Apolloni and C.~Gentile.
\newblock Sample size lower bounds in {PAC} learning by algorithmic complexity
  theory.
\newblock {\em Theoretical Computer Science}, 209:141--162, 1998.

\bibitem[AL88]{angluin:randomclassificationnoise}
D.~Angluin and P.~Laird.
\newblock Learning from noisy examples.
\newblock {\em Machine Learning}, 2(4):343--370, 1988.

\bibitem[AM14]{ambainis&montanaro:wildcard&CGT}
A.~Ambainis and A.~Montanaro.
\newblock Quantum algorithms for search with wildcards and combinatorial group
  testing.
\newblock {\em Quantum Information {\&} Computation}, 14(5-6):439--453, 2014.
\newblock arXiv:1210.1148.

\bibitem[AS05]{atici&servedio:qlearning}
A.~{At\i c\i} and R.~Servedio.
\newblock Improved bounds on quantum learning algorithms.
\newblock {\em Quantum Information Processing}, 4(5):355--386, 2005.
\newblock quant-ph/0411140.

\bibitem[AS09]{atici&servedio:testing}
A.~{At\i c\i} and R.~Servedio.
\newblock Quantum algorithms for learning and testing juntas.
\newblock {\em Quantum Information Processing}, 6(5):323--348, 2009.
\newblock arXiv:0707.3479.

\bibitem[Aud08]{audibert:agnosticconstant2}
J.~Audibert.
\newblock Fast learning rates in statistical inference through aggregation,
  2008.
\newblock Research Report 06-20, Certis—Ecole des Ponts. math/0703854.

\bibitem[Aud09]{audibert:agnosticconstant1}
J.~Audibert.
\newblock Fast learning rates in statistical inference through aggregation.
\newblock {\em The Annals of Statistics}, 37(4):1591--1646, 2009.
\newblock arXiv:0909.1468v1.

\bibitem[BCD06]{bacon:hiddensubgroup}
D.~Bacon, A.~Childs, and W.~{van} Dam.
\newblock Optimal measurements for the dihedral hidden subgroup problem.
\newblock {\em Chicago Journal of Theoretical Computer Science}, 2006.
\newblock Earlier version in FOCS'05. quant-ph/0504083.

\bibitem[BCWW01]{bcww:fp}
H.~Buhrman, R.~Cleve, J.~Watrous, and R.~{de} Wolf.
\newblock Quantum fingerprinting.
\newblock {\em Physical Review Letters}, 87(16), 2001.
\newblock quant-ph/0102001.

\bibitem[BEHW89]{blumer:optimalpacupper}
A.~Blumer, A.~Ehrenfeucht, D.~Haussler, and M.~K. Warmuth.
\newblock Learnability and the {V}apnik-{C}hervonenkis dimension.
\newblock {\em Journal of the ACM}, 36(4):929--965, 1989.

\bibitem[BJ99]{bshouty:quantumpac}
N.~H. Bshouty and J.~C. Jackson.
\newblock Learning {DNF} over the uniform distribution using a quantum example
  oracle.
\newblock {\em SIAM Journal on Computing}, 28(3):1136–--1153, 1999.
\newblock Earlier version in COLT'95.

\bibitem[BK02]{barnum:pgmmeasurement}
H.~Barnum and E.~Knill.
\newblock Reversing quantum dynamics with near-optimal quantum and classical
  fidelity.
\newblock {\em Journal of Mathematical Physics}, 43:2097--2106, 2002.
\newblock quant-ph/0004088.

\bibitem[BV97]{bernstein&vazirani:qcomplexity}
E.~Bernstein and U.~Vazirani.
\newblock Quantum complexity theory.
\newblock {\em SIAM Journal on Computing}, 26(5):1411--1473, 1997.
\newblock Earlier version in STOC'93.

\bibitem[DS16]{daniely&shalevshwartz:limitdnf}
A.~Daniely and S.~{Shalev-Shwartz}.
\newblock Complexity theoretic limitations on learning {DNF}'s.
\newblock In {\em Proceedings of the 29th Conference on Learning Theory
  (COLT'16)}, 2016.

\bibitem[EF01]{eldar&forney:squarerootmeasurement}
Y.~C. Eldar and G.~D. {Forney Jr}.
\newblock On quantum detection and the square-root measurement.
\newblock {\em IEEE Transactions and Information Theory}, 47(3):858--872, 2001.
\newblock quant-ph/0005132.

\bibitem[EHKV89]{ehrenfeucht:lowerboundforlearning}
A.~Ehrenfeucht, D.~Haussler, M.~J. Kearns, and L.~G. Valiant.
\newblock A general lower bound on the number of examples needed for learning.
\newblock {\em Information and Computation}, 82(3):247--261, 1989.
\newblock Earlier version in COLT'98.

\bibitem[EMV03]{eldar:optimalquantumdetectors}
Y.~C. Eldar, A.~Megretski, and G.~C. Verghese.
\newblock Designing optimal quantum detectors via semidefinite programming.
\newblock {\em {IEEE} Transactions Information Theory}, 49(4):1007--1012, 2003.
\newblock quant-ph/0205178.

\bibitem[Gav12]{gavinsky:predictivelearning}
D.~Gavinsky.
\newblock Quantum predictive learning and communication complexity with single
  input.
\newblock {\em Quantum Information and Computation}, 12(7-8):575--588, 2012.
\newblock Earlier version in COLT'10. arXiv:0812.3429.

\bibitem[GH01]{gentile&helmbold:it}
C.~Gentile and D.~P. Helmbold.
\newblock Improved lower bounds for learning from noisy examples: An
  information-theoretic approach.
\newblock {\em Information and Computation}, 166:133--155, 2001.

\bibitem[Gro96]{grover:search}
L.~K. Grover.
\newblock A fast quantum mechanical algorithm for database search.
\newblock In {\em Proceedings of 28th ACM STOC}, pages 212--219, 1996.
\newblock quant-ph/9605043.

\bibitem[Han16]{hanneke:optimalpaclower}
S.~Hanneke.
\newblock The optimal sample complexity of {PAC} learning.
\newblock {\em Journal of Machine Learning Research}, 17(38):1--15, 2016.
\newblock arXiv:1507.00473.

\bibitem[Hau92]{haussler:agnosticlearning}
D.~Haussler.
\newblock Decision theoretic generalizations of the {PAC} model for neural net
  and other learning applications.
\newblock {\em Information and Computation}, 100(1):78–--150, 1992.

\bibitem[HHL09]{hhl:lineq}
A.~Harrow, A.~Hassidim, and S.~Lloyd.
\newblock Quantum algorithm for solving linear systems of equations.
\newblock 103(15):150502, 2009.
\newblock arXiv:0811.3171.

\bibitem[HJS{\etalchar{+}}96]{hausladen:squareroot}
P.~Hausladen, R.~Jozsa, B.~Schumacher, M.~Westmoreland, and W.~K. Wootters.
\newblock Classical information capacity of a quantum channel.
\newblock {\em Physical Review A}, 54:1869--1876, 1996.

\bibitem[HMP{\etalchar{+}}10]{hunziker:quantumexactlearning}
M.~Hunziker, D.~A. Meyer, J.~Park, J.~Pommersheim, and M.~Rothstein.
\newblock The geometry of quantum learning.
\newblock {\em Quantum Information Processing}, 9(3):321--341, 2010.
\newblock quant-ph/0309059.

\bibitem[HW94]{hausladen:pgmintroduction}
P.~Hausladen and W.~K. Wootters.
\newblock A ‘pretty good’ measurement for distinguishing quantum states.
\newblock {\em Journal Of Modern Optics}, 41:2385--2390, 1994.

\bibitem[Jac97]{jackson:dnf}
J.~C. Jackson.
\newblock An efficient membership-query algorithm for learning {DNF} with
  respect to the uniform distribution.
\newblock {\em Journal of Computer and System Sciences}, 55(3):414--440, 1997.
\newblock Earlier version in FOCS'94.

\bibitem[JTY02]{jackson:quantumdnf}
J.~C. Jackson, C.~Tamon, and T.~Yamakami.
\newblock Quantum {DNF} learnability revisited.
\newblock In {\em Proceedings of 8th COCOON}, pages 595--604, 2002.
\newblock quant-ph/0202066.

\bibitem[JZ09]{jain&zhang:newoneway}
R.~Jain and S.~Zhang.
\newblock New bounds on classical and quantum one-way communication complexity.
\newblock {\em Theoretical Computer Science}, 410(26):2463--2477, 2009.
\newblock arXiv:0802.4101.

\bibitem[Kot14]{kothari:oracleidentification}
R.~Kothari.
\newblock An optimal quantum algorithm for the oracle identification problem.
\newblock In {\em 31st International Symposium on Theoretical Aspects of
  Computer Science (STACS 2014)}, pages 482--493, 2014.
\newblock arXiv:1311.7685.

\bibitem[KP16]{aryeh:exactconstantagnostic}
A.~Kontorovich and I.~Pinelis.
\newblock Exact lower bounds for the agnostic probably-approximately-correct
  {(PAC)} machine learning model, 2016.
\newblock Preprint at arxiv:1606.08920.

\bibitem[KSS94]{kearns:agnosticlearning}
M.~J. Kearns, R.~E. Schapire, and L.~Sellie.
\newblock Toward efficient agnostic learning.
\newblock {\em Machine Learning}, 17(2-3):115--141, 1994.
\newblock Earlier version in COLT'92.

\bibitem[KV94a]{kearns&valiant:blum}
M.~J. Kearns and L.~G. Valiant.
\newblock Cryptographic limitations on learning {B}oolean formulae and finite
  automata.
\newblock {\em Journal of the ACM}, 41(1):67--95, 1994.

\bibitem[KV94b]{kearnss&vazirani:learningbook}
M.~J. Kearns and U.~V. Vazirani.
\newblock {\em An introduction to computational learning theory}.
\newblock MIT Press, 1994.

\bibitem[Mon07]{montanaro:distinguishability}
A.~Montanaro.
\newblock On the distinguishability of random quantum states.
\newblock {\em Communications in Mathematical Physics}, 273(3):619--636, 2007.
\newblock quant-ph/0607011.

\bibitem[Mon12]{montanaro:learningpolynomials}
A.~Montanaro.
\newblock The quantum query complexity of learning multilinear polynomials.
\newblock {\em Information Processing Letters}, 112(11):438--442, 2012.
\newblock arXiv:1105.3310.

\bibitem[{O'D}14]{odonnell:analysis}
R.~{O'Donnell}.
\newblock {\em Analysis of Boolean Functions}.
\newblock Cambridge University Press, 2014.

\bibitem[SB14]{shwartz&david:learningbook}
S.~{Shalev-Shwartz} and S.~{Ben-David}.
\newblock {\em Understanding machine learning: From theory to algorithms}.
\newblock Cambridge University Press, 2014.

\bibitem[SG04]{servedio&gortler:equivalencequantumclassical}
R.~Servedio and S.~Gortler.
\newblock Equivalences and separations between quantum and classical
  learnability.
\newblock {\em SIAM Journal on Computing}, 33(5):1067--1092, 2004.
\newblock Combines earlier papers from ICALP'01 and CCC'01. quant-ph/0007036.

\bibitem[Sim96]{simon:agnosticlowerbound}
H.~U. Simon.
\newblock General bounds on the number of examples needed for learning
  probabilistic concepts.
\newblock {\em Journal of Computer and System Sciences}, 52(2):239--254, 1996.
\newblock Earlier version in COLT'93.

\bibitem[Sim15]{simon:almostoptimalpac}
H.~U. Simon.
\newblock An almost optimal {PAC} algorithm.
\newblock In {\em Proceedings of the 28th Conference on Learning Theory
  (COLT)}, pages 1552--1563, 2015.

\bibitem[Tal94]{talagrand:agnosticupperbound}
M.~Talagrand.
\newblock Sharper bounds for {G}aussian and empirical processes.
\newblock {\em The Annals of Probability}, pages 28--76, 1994.

\bibitem[Val84]{valiant:paclearning}
L.~Valiant.
\newblock A theory of the learnable.
\newblock {\em Communications of the ACM}, 27(11):1134–--1142, 1984.

\bibitem[VC71]{vapnik:vcdimension}
V.~Vapnik and A.~Chervonenkis.
\newblock On the uniform convergence of relative frequencies of events to their
  probabilities.
\newblock {\em Theory of Probability \& Its Applications}, 16(2):264--280,
  1971.

\bibitem[VC74]{vapnik:agnosticlowerbound}
V.~Vapnik and A.~Chervonenkis.
\newblock Theory of pattern recognition.
\newblock 1974.
\newblock In Russian.

\bibitem[Ver90]{verbeurgt:learningdnf}
K.~A. Verbeurgt.
\newblock Learning {DNF} under the uniform distribution in quasi-polynomial
  time.
\newblock In {\em Proceedings of the 3rd Annual Workshop on Computational
  Learning Theory (COLT'90)}, pages 314--326, 1990.

\bibitem[WKS14]{wiebe:quantumdeeplearning}
N.~Wiebe, A.~Kapoor, and K.~M. Svore.
\newblock Quantum deep learning, 2014.
\newblock Preprint at arXiv:1412.3489.

\bibitem[WKS16]{wiebe:quantumperceptronmodels}
N.~Wiebe, A.~Kapoor, and K.~M. Svore.
\newblock Quantum perceptron models, 2016.
\newblock Preprint at arXiv:1602.04799.

\bibitem[Zha10]{zhang:improvedvcbound}
C.~Zhang.
\newblock An improved lower bound on query complexity for quantum {PAC}
  learning.
\newblock {\em Information Processing Letters}, 111(1):40--45, 2010.

\end{thebibliography}

\newcommand{\etalchar}[1]{$^{#1}$}

\end{document}